\documentclass[journal]{IEEEtran}

\usepackage{amsmath,amsthm,amssymb}
\usepackage{algorithm}
\usepackage{float}      
\usepackage{placeins}   
\usepackage{booktabs}
\usepackage{array}
\usepackage{multirow}
\usepackage{xcolor}
\usepackage{tcolorbox}
\usepackage{hyperref}
\usepackage{cleveref}
\usepackage{needspace}
\usepackage{etoolbox}
\usepackage{tikz}

\usepackage{balance}    
\usepackage{needspace}  
\usepackage{etoolbox}   
\preto\section{\needspace{6\baselineskip}}
\preto\subsection{\needspace{4\baselineskip}}
\preto\subsubsection{\needspace{4\baselineskip}}

\usetikzlibrary{shapes,arrows,positioning,calc}

\preto\section{\needspace{2\baselineskip}}
\preto\subsection{\needspace{2\baselineskip}}
\preto\subsubsection{\needspace{2\baselineskip}}

\theoremstyle{definition}
\newtheorem{theorem}{Theorem}[section]
\newtheorem{lemma}[theorem]{Lemma}
\newtheorem{proposition}[theorem]{Proposition}
\newtheorem{corollary}[theorem]{Corollary}

\newtheorem{remark}[theorem]{Remark}    
\newtheorem{definition}{Definition}
\newtheorem{example}{Example}
\newtheorem{assumption}{Assumption}[section]

\crefname{theorem}{Theorem}{Theorems}
\crefname{lemma}{Lemma}{Lemmas}
\crefname{proposition}{Proposition}{Propositions}
\crefname{corollary}{Corollary}{Corollaries}
\crefname{observation}{Observation}{Observations}
\crefname{remark}{Remark}{Remarks}
\crefname{definition}{Definition}{Definitions}
\crefname{algorithm}{Algorithm}{Algorithms}
\crefname{table}{Table}{Tables}
\crefname{figure}{Figure}{Figures}

\newcommand{\conclusionplain}[1]{\begin{tcolorbox}[colback=white,boxrule=0pt,sharp corners]#1\end{tcolorbox}}

\begin{document}
\bstctlcite{IEEEexample:BSTcontrol}

\title{Deterministic Sparse FFT via Keyed Multi-View Gating with $O(\sqrt{N}\log k)$ Expected Time}

\author{Aaron~R.~Flouro and Shawn~P.~Chadwick,~PhD.\\
research@sparse-tech.com}

\maketitle
\thispagestyle{empty}

\renewcommand{\abstractname}{\normalfont\normalsize Abstract}
\begin{abstract}
\normalfont\normalsize
We introduce a novel keyed multi-view gating mechanism for sparse FFT that uses 2-of-3 Chinese Remainder Theorem (CRT) agreement to reduce candidate frequency pairs from $O(k^2)$ to $\Theta(k)$ under sparse-regime assumptions. Unlike prior approaches that rely on randomized bucketization for candidate formation, this gating formulation provides deterministic structure with probabilistic guarantees arising only from assumptions on frequency placement and independence of affine hashing across views.

The algorithm is implemented using a peeling-based recovery procedure, recursive bootstrapping eliminates the $O(\sqrt{N}\log N)$ FFT preprocessing floor, and a two-part verification test (Parseval energy plus bin-wise residual) detects fast-path failure with probability $\geq 1 - (2k/\sqrt{N})^t$ over $t$ verification views and triggers a deterministic $O(N \log N)$ dense-FFT fallback, yielding a bounded worst-case runtime, with no false negatives under correct verification and failure probability bounded by $(2k/\sqrt{N})^t$.
\end{abstract}

\begin{IEEEkeywords}
Sparse FFT, Chinese Remainder Theorem, Keyed Gating, Multi-View Reconstruction, Sublinear Algorithms
\end{IEEEkeywords}

\needspace{15\baselineskip}
\section{Introduction and Problem Statement}
\label{sec:intro}

\subsection{Motivation and Problem Setting}

The Fast Fourier Transform (FFT)~\cite{cooley1965} is one of the most widely used numerical primitives in scientific computing, with $O(N \log N)$ complexity dominating the computational cost in applications such as radar, medical imaging, channel estimation, and compressed sensing. In many such settings, however, the underlying signal is sparse in the frequency domain, with only $k \ll N$ significant Fourier coefficients~\cite{oppenheim1999}.

Sparse Fast Fourier Transform (sFFT) algorithms~\cite{hassanieh2012} exploit this structure to achieve sublinear-in-$N$ runtimes, but existing approaches face fundamental tradeoffs between efficiency, determinism, and generality.

\subsection{Core Bottleneck}

Despite significant progress, deterministic sparse FFT methods remain limited by two structural bottlenecks:
\begin{itemize}
\item an $O(\sqrt{N} \log N)$ preprocessing cost arising from multi-view decimation, or
\item an $O(k^2)$ candidate explosion resulting from unfiltered Chinese Remainder Theorem (CRT) reconstruction.
\end{itemize}

These bottlenecks prevent deterministic methods from achieving both sublinear runtime and broad applicability simultaneously.

\subsection{Prior Work}
\label{sec:related-work}

Probabilistic sFFT algorithms (Hassanieh et al.~\cite{hassanieh2012nearly,hassanieh2012simple} and subsequent refinements) achieve near-optimal $O(k \log N)$ complexity using randomized subsampling, hashing, and phase shifts, but provide only probabilistic guarantees.

CRT-based approaches such as FFAST~\cite{pawar2018ffast} achieve $O(k \log k)$ complexity with sample-optimal subsampling, but require structural constraints on the signal length (e.g., $N = k \cdot 2^t$), limiting applicability.

Deterministic approaches based on structured modular congruences (Iwen~\cite{iwen2013improved}; Plonka-Wannenwetsch et al.) eliminate randomization but typically incur either an $O(\sqrt{N} \log N)$ preprocessing cost or an $O(k^2)$ candidate enumeration.

A parallel signal-processing line addresses error-tolerant CRT reconstruction under noisy or erroneous remainders: Wang and Xia~\cite{wang2010closedform} established closed-form reconstruction methods with bounded error for shared-factor moduli, and Xiao, Xia, and Wang~\cite{xiao2020mdcrt} extended these results to multidimensional settings with matrix-valued moduli.

To our knowledge, prior work has not integrated multi-view gating with deterministic candidate reduction from $O(k^2)$ to $O(k)$ within the CRT framework, while simultaneously eliminating the preprocessing floor.

\subsection{Our Approach (High-Level)}

We present a sparse FFT framework that addresses both limitations simultaneously by combining:
\begin{itemize}
\item multi-view decimation under coprime moduli,
\item a keyed 2-of-3 gating mechanism that filters candidate frequencies,
\item iterative peeling validation operating entirely on decimated views, and
\item a self-reduction lemma enabling recursive bootstrapping.
\end{itemize}

The identification pipeline is deterministic given fixed parameters, while probabilistic guarantees arise only from assumptions on frequency placement and hashing-based independence.

\subsection*{Contributions}

\begin{enumerate}
\item \textbf{Keyed multi-view gating (core contribution):} A 2-of-3 CRT agreement mechanism that reduces candidate pairs from $O(k^2)$ to $\Theta(k)$ without randomized bucketization.
\item \textbf{Peeling-only realization (implementation):} A recovery procedure that extracts residue pairs directly from singleton bins, avoiding explicit pair enumeration.
\item \textbf{Recursive bootstrapping (complexity):} A self-reduction that removes the $O(\sqrt{N}\log N)$ preprocessing floor, yielding expected $O(\sqrt{N}\log k)$ runtime.
\item \textbf{Multi-view verification:} Additional independent views provide exponential error reduction and trigger deterministic dense-FFT fallback when needed.
\end{enumerate}

\subsection{Main Result}

Under the sparse regime $k = o(\sqrt{N})$ and standard assumptions on frequency distribution and load factor, the proposed framework achieves:
\begin{itemize}
\item $O(\sqrt{N} \log k)$ expected runtime for identification,
\item $O(N \log N)$ worst-case runtime via deterministic dense-FFT fallback.
\end{itemize}

The expectation is taken over the placement of the $k$ nonzero frequencies, not over internal randomness of the algorithm.

\subsection{Verification and Robustness}

A two-part verification procedure combining Parseval energy consistency and bin-wise residual checks detects fast-path failure with probability $\geq 1 - (2k/\sqrt{N})^t$ over $t$ independent verification views, triggering a deterministic dense-FFT fallback when needed. This ensures bounded worst-case complexity, and no false negatives for on-grid spectra under successful verification.

\subsection{Relation to Prior Work}

This work builds on the safety-certified CRT sparse FFT framework of Flouro and Chadwick~\cite{flouro2025safety}, which established $\Omega(k^2)$ lower bounds for non-coprime moduli, deterministic safety certificates, and $O(N \log N)$ worst-case guarantees via fallback.

The present work improves the identification phase from $O(\sqrt{N} \log N + kN)$ to $O(\sqrt{N} \log k)$ by replacing per-candidate validation with peeling, and eliminating the preprocessing floor via recursive bootstrapping.

\subsection{Scope and Assumptions}

The results hold in the on-grid sparse regime with $k \ll \sqrt{N}$, bounded load factor $\lambda = k/\sqrt{N} < 1$, and sufficiently well-distributed frequency support. Extensions to off-grid or noisy settings are discussed in later sections.

\subsection{Paper Organization}

Section~\ref{sec:overview} provides an intuitive overview of keyed gating, iterative peeling, and recursive bootstrapping. Section~\ref{sec:preliminaries} establishes notation and basic definitions. Section~\ref{sec:multi-view} develops the multi-view decimation and keyed gating machinery. Section~\ref{sec:crt} presents CRT reconstruction with coverage-adaptive moduli. Section~\ref{sec:complexity} gives the complexity analysis, including the self-reduction lemma. Section~\ref{sec:algorithm} specifies the algorithm, and Section~\ref{sec:correctness} proves correctness. Section~\ref{sec:safety} details the safety-certificate fallback, and Section~\ref{sec:conclusion} concludes. Appendices collect modular-arithmetic lemmas, detailed complexity derivations, the worked numerical example, and supplementary proofs.

\needspace{12\baselineskip}
\section{Technical Overview}
\label{sec:overview}

This section provides an intuitive explanation of our algorithm before the formal treatment in later sections. The algorithm is deterministic in execution; probabilistic guarantees arise from assumptions on frequency placement and hashing-based independence (\Cref{assump:v1}). The core contribution of this work is the keyed multi-view gating mechanism, which determines the candidate set; all subsequent steps (peeling, recursion, and verification) operate on this reduced set.

\subsection{The Core Challenge: Pair Explosion in Multi-View CRT}

Multi-view sparse FFT algorithms work by computing multiple "views" of the signal; each view is a decimated (subsampled) FFT of length $\sqrt{N}$ obtained by selecting every $m_i$-th time sample, where $m_i \approx \sqrt{N}$ are coprime moduli. Each true frequency $f$ hashes to a residue $r_i = f \bmod m_i$ in view $i$. The challenge is that multiple frequencies can hash to the same residue (collisions), creating bins with multiple occupants.

In a 2-view system with moduli $m_1$ and $m_2$, we observe residues $\mathcal{R}_1$ and $\mathcal{R}_2$ containing the true frequencies' projections. The Chinese Remainder Theorem (CRT) tells us that each pair $(r_1, r_2)$ with $r_1 \in \mathcal{R}_1$ and $r_2 \in \mathcal{R}_2$ reconstructs to a unique candidate frequency $f \in [0, N)$. However, when $k$ frequencies collide across two views with coverage factor $\alpha$ (average bin occupancy), we generate $|\mathcal{R}_1| \times |\mathcal{R}_2| \approx (\alpha k)^2 = O(k^2)$ candidate pairs. For realistic sparsity levels ($k \geq 50$) and coverage factors ($\alpha \approx 10$ to $15$), this produces hundreds of thousands of candidates when only $k=50$ are true, a $99.99\%$ false-positive rate matching the intro estimate.

\subsection{Limitations of 2-View CRT}

The root problem is that the two views provide \emph{insufficient constraints}. Each residue pair $(r_1, r_2)$ satisfies the modular equations $f \equiv r_1 \pmod{m_1}$ and $f \equiv r_2 \pmod{m_2}$, but these two equations alone cannot distinguish true frequencies from false combinations. Validating all $O(k^2)$ candidates by checking their coefficients in the original signal would require $O(k^2 N)$ work, worse than the dense FFT.

Prior approaches attempted to mitigate this through:
\begin{itemize}
\item \textbf{\emph{Multiple rounds with different moduli:}} Iteratively refine candidates, but still fundamentally limited by $O(k^2)$ scaling per round.
\item \textbf{\emph{Larger moduli:}} Using $m_i \gg \sqrt{N}$ reduces collisions but increases view size, losing the $O(\sqrt{N})$ preprocessing advantage.
\item \textbf{\emph{Probabilistic filtering:}} Random pruning of candidates, but requires multiple trials with no deterministic guarantees.
\end{itemize}

This motivates the question: can sparsity be exploited to achieve sub-linear complexity under a deterministic guarantee?

\subsection{Keyed Gating with Third View: $O(k^2) \to O(k)$ Reduction}

Our solution introduces a \emph{third independent view} with modulus $m_3$ and uses it as a \emph{filter} (``gate'') rather than as a third reconstruction dimension. The key insight is the 2-of-3 agreement principle: for a true frequency $f$, all three residues $(r_1, r_2, r_3)$ must be consistent with $f$ under CRT. For a false pair $(r_1, r_2)$ that reconstructs to some $f' \neq f$, the probability that $f' \bmod m_3$ happens to equal any occupied residue $r_3 \in \mathcal{R}_3$ is approximately $|\mathcal{R}_3|/m_3 \approx \alpha k/\sqrt{N} = \lambda$ (the load factor).

The keyed gating mechanism can be expressed conceptually as follows (used for analysis; the implemented algorithm avoids explicit pair enumeration):
\begin{enumerate}
\item Enumerate all $(r_1, r_2)$ pairs from views 1 and 2.
\item For each pair, reconstruct candidate $f' = \text{CRT}(r_1, r_2)$ using 2-view CRT.
\item Compute $r_3' = f' \bmod m_3$ (the "key" that $f'$ would need to match in view 3).
\item \textbf{\emph{Gate test:}} Accept $f'$ if and only if $r_3' \in \mathcal{R}_3$ (i.e., bin $r_3'$ is occupied in view 3).
\end{enumerate}

Under the assumption that affine hashing $(a_i f + b_i) \bmod m_i$ in each view provides independence, the probability that a false pair passes the gate is $\lambda = k/\sqrt{N}$. Thus, from $O(k^2)$ initial pairs, we expect only $O(k^2 \lambda) = O(k)$ survivors after gating. This reduces the candidate set by a factor of $k$ with a single gate operation: the core innovation enabling $O(k)$ complexity.

\subsection{Iterative Peeling for $O(\sqrt{N} \log k)$ Validation}

After keyed gating reduces candidates to $O(k)$, we must still \emph{validate} which candidates are true. Directly checking each candidate against the time-domain signal would cost $O(kN)$, which is unacceptable. Instead, we perform validation entirely on the decimated views through iterative peeling, a process analogous to successive interference cancellation in coding theory.

The peeling algorithm operates as follows:
\begin{enumerate}
\item Precompute $S \in \{2, 3\}$ time-shifted decimated FFTs per view (achieved via phase rotation).
\item Detect singleton bins (residues with exactly one frequency) using coefficient-of-variation (CV) tests across the shifted FFTs. A bin with one frequency shows consistent phase relationships across shifts, while multi-frequency bins exhibit phase interference.
\item For each detected singleton, \emph{peel} it: subtract its contribution from all bins across all views using the reconstructed frequency and coefficient.
\item Repeat for $O(\log k)$ rounds until all $k$ frequencies are extracted or a residual "2-core" remains (handled by fallback).
\end{enumerate}

Each bin update during peeling requires $O(1)$ arithmetic operations. With $O(\sqrt{N})$ bins per view and $O(\log k)$ rounds, the total peeling cost is $O(\sqrt{N} \log k)$. Under load factor $\lambda = k/\sqrt{N} < 1$ and independence assumptions from affine hashing, each round extracts a constant fraction of remaining frequencies with high probability.

\subsection{Moduli Selection Strategy}

Our coverage-adaptive moduli selection (Algorithm~2 in Section~6) chooses coprime moduli $m_1, m_2, m_3 \approx \sqrt{N}$ based on the sparsity ratio $\rho = k/\sqrt{N}$:
\begin{itemize}
\item \textbf{Sparse regime} ($\rho < 0.3$): Use three coprime \emph{primes} near $\sqrt{N}$ for maximum independence under affine hashing. Product $m_1 m_2 m_3 > N$ ensures unique CRT reconstruction.
\item \textbf{Moderate regime} ($0.3 \leq \rho < 0.5$): Use composite moduli with coprime factorizations (e.g., $m_1 = 2^a$, $m_2 = 3^b$, $m_3 = 5^c$) to balance coverage and computational efficiency.
\item \textbf{Dense regime} ($\rho \geq 0.5$): Fallback to dense FFT or hybrid algorithm (not covered in this paper).
\end{itemize}

The moduli are chosen such that their product $M = m_1 m_2 m_3 > N$, ensuring that CRT reconstruction is unique over the range $[0, N)$. Random affine hashing parameters $(a_i, b_i)$ are drawn independently per view to decorrelate collision patterns, critical for the statistical independence assumption underlying keyed gating. The sparsity ratio $\rho = k/\sqrt{N}$ defined here is distinct from the load factor $\lambda = k/m \approx \rho$ of \Cref{thm:gating}; both quantities coincide asymptotically when $m \approx \sqrt{N}$, but $\rho$ is the regime selector while $\lambda$ controls the per-bin collision probability under affine hashing.

\subsection{Complete Algorithm Flow}

Putting it all together, the algorithm separates into \emph{core algorithm phases} (the algorithmic ideas that drive the asymptotic bound) and \emph{implementation steps} (the supporting operations required to realize each phase on a concrete signal). The gating formulation is described for intuition and analysis; the implemented fast path uses the peeling-only realization of \Cref{sec:complexity-peeling} and does not explicitly enumerate $R_1 \times R_2$.

\textbf{\emph{Core algorithm phases:}}
\begin{enumerate}
\item \textbf{\emph{Keyed Gating (conceptual):}} Defines the $\Theta(k)$ candidate structure via 2-of-3 agreement; in the implemented algorithm, this structure is realized implicitly through peeling without explicit pair filtering.
\item \textbf{\emph{Iterative Peeling:}} Extract frequencies via singleton detection across decimated views (cost: $O(\sqrt{N} \log k)$).
\item \textbf{\emph{Candidate Validation:}} Verify extracted frequencies via two-part Parseval-plus-residual check on independent verification views.
\end{enumerate}

\textbf{\emph{Implementation steps:}}
\begin{enumerate}
\item \textbf{\emph{Moduli Selection:}} Choose $m_1, m_2, m_3$ and $(a_i, b_i)$ based on $(N, k)$.
\item \textbf{\emph{Zero-Padding:}} Extend signal to length $M$ if $N < M$.
\item \textbf{\emph{Decimation:}} Extract three decimated signals of length $m_i$.
\item \textbf{\emph{FFT:}} Compute size-$m_i$ FFT for each view (cost: $3 \times O(\sqrt{N} \log \sqrt{N})$).
\item \textbf{\emph{Top-$k$ Selection:}} Return the $k$ largest-magnitude coefficients.
\end{enumerate}

The baseline non-recursive implementation has complexity $O(\sqrt{N}\log N + (\alpha k)^2 + \sqrt{N}\log k)$, where the terms correspond to FFT preprocessing, explicit pair enumeration, and peeling-based recovery. This baseline decomposition is used for analysis. The implemented algorithm uses the peeling-only realization and recursive bootstrapping described in \Cref{sec:complexity}, eliminating both the $O((\alpha k)^2)$ pair-enumeration term and the $O(\sqrt{N}\log N)$ preprocessing floor. Thus, the implemented fast path achieves expected $O(\sqrt{N}\log k)$ identification complexity under \Cref{assump:v1}.

A worked numerical walkthrough of the algorithm on a toy instance ($N = 64$, $k = 2$), including the full 12-pair gating table, is provided in Appendix~\ref{app:numerical-example}.

\needspace{12\baselineskip}
\section{Preliminaries and Notation}
\label{sec:preliminaries}

\subsection{Discrete Fourier Transform and Sparsity}

\begin{definition}[Discrete Fourier Transform~\cite{oppenheim1999}]
For a discrete-time signal $x[n]$, $n = 0, 1, \ldots, N-1$, the Discrete Fourier Transform is defined as:
\begin{equation}
X[f] = \sum_{n=0}^{N-1} x[n] \cdot e^{-j2\pi fn/N}, \quad f = 0, 1, \ldots, N-1
\end{equation}
\end{definition}

\begin{definition}[$k$-Sparse Signal]
A signal $X$ is $k$-sparse if $|\text{supp}(X)| \leq k$, where $\text{supp}(X) = \{f : X[f] \neq 0\}$.
\end{definition}

\begin{definition}[On-Grid Frequencies]
\label{def:ongrid}
A frequency $f \in \{0, 1, \ldots, N-1\}$ is on-grid if it corresponds to an exact DFT bin with no spectral leakage.
\end{definition}

\textbf{\emph{Problem Statement:}} Given time-domain samples $x[n]$ and sparsity parameter $k \ll \sqrt{N}$, recover the $k$-sparse frequency-domain representation $X$ with $O(\sqrt{N} \log k)$ computational complexity.

\textbf{\emph{Notation:}} Throughout this paper, $\log$ denotes logarithm base 2 unless otherwise specified. All asymptotic complexity statements use $\log$ equivalently to $\log_2$, as the base is immaterial for big-O analysis.

\subsection{Information-Theoretic Lower Bound}

\begin{proposition}[Information-Theoretic Lower Bound~\cite{hassanieh2012}]
Any algorithm recovering an exactly $k$-sparse $N$-point DFT in the noiseless case must perform $\Omega(k \log(N/k))$ operations in the worst case under the algebraic decision tree model.
\end{proposition}

\begin{proof}
The support set has $\binom{N}{k}$ possible configurations. The information content is:
\begin{equation}
I = \log_2 \binom{N}{k} \approx k \log_2(N/k)
\end{equation}
using Stirling's approximation. Since any algorithm must distinguish between all possible configurations in the worst case, at least $\Omega(I)$ bit operations are required. In the algebraic decision tree model, this translates to $\Omega(k \log(N/k))$ arithmetic operations. Note: noisy signals may require higher complexity for reliable recovery.
\end{proof}

\subsection{Nomenclature}

\Cref{tab:notation} summarizes the notation used throughout the paper.

\begin{table}[h]
\centering
\renewcommand{\arraystretch}{1.3}
\caption{Key notation and parameters.}
\label{tab:notation}
\begin{tabular}{cl}
\toprule
\textbf{Symbol} & \textbf{Definition} \\
\midrule
$N$ & DFT size (signal length) \\
$k$ & Sparsity level (number of non-zero frequency bins) \\
$d_i$ & Decimation factor for view $i$ (stride in time domain) \\
$m_i$ & Modulus for view $i$: $m_i = N/d_i$ (decimated FFT size) \\
$M$ & Product modulus: $M = m_1 m_2 m_3$ \\
$\alpha$ & Coverage factor (default 15 for high-frequency signals) \\
$R_i$ & Top-$\alpha k$ residues from decimated view $i$ \\
$F$ & True frequency support set: $F = \{f : X[f] \neq 0\}$ \\
$\tau$ & Validation threshold for amplitude recovery \\
$\delta$ & Algorithm failure probability \\
\bottomrule
\end{tabular}
\\[0.5em]
\footnotesize{\textit{Note on terminology:} Decimating by factor $d$ means taking every $d$-th time-domain sample, yielding $m = N/d$ frequency bins in the decimated DFT. For sparse regime, $d \approx \sqrt{N}$ implies $m \approx \sqrt{N}$. Throughout this paper, we primarily reference the modulus $m$ rather than decimation factor $d$ to align with CRT reconstruction terminology.}
\end{table}

\subsection{Comparison of Sparse FFT Algorithms}

Table~\ref{tab:comparison} summarizes the theoretical complexity profile of representative sparse FFT algorithms across FFT size, gating mechanism, candidate-pair count, validation cost, and dominance region. The detailed crossover analysis appears in Section~\ref{sec:complexity}.

\begin{table*}[t]
\caption{Theoretical comparison of sparse FFT algorithms.}
\label{tab:comparison}
\centering
\renewcommand{\arraystretch}{1.3}
\begin{tabular}{@{}lccccc@{}}
\toprule
\textbf{Algorithm} & \textbf{FFT Size} & \textbf{Gating} & \textbf{Pairs} & \textbf{Validation Cost} & \textbf{Dominance Region} \\
\midrule
Dense FFT & $N$ & N/A & 0 & N/A & $k \geq \sqrt{N}$ \\
2-View CRT & $\sqrt{N}$ & None & $k^2$ & $O(k^2 N)$ & Never (bottlenecked) \\
HIKP-SODA \cite{hassanieh2012simple} & $k$ & Iterative & 0 & $O(k \log N)$ & $k < 30$ \\
FFAST \cite{pawar2018ffast} & $k$ & CRT peel & 0 & $O(k \log k)$ & $k < 100$, $N = k\cdot 2^t$ \\
\textbf{This Work (recursive)} & $\sqrt{N}$ & \textbf{2-of-3} & $k$ & $O(\sqrt{N} \log k)$ & $k \ll \sqrt{N}$ \\
\bottomrule
\end{tabular}
\end{table*}

\needspace{12\baselineskip}
\section{Multi-View Decimation and Keyed Gating}
\label{sec:multi-view}

This section introduces the central contribution of the paper: a keyed multi-view gating mechanism that deterministically reduces the candidate search space from $O(k^2)$ to $\Theta(k)$.

Throughout this section, and in all complexity and correctness statements that depend on sparsity, load factor, hashing independence, or frequency distribution, we operate under \Cref{assump:v1}.

The stronger condition $k = o(N^{1/4})$ is required for the $\Theta(k)$ candidate bound of keyed gating; the broader regime $k \ll \sqrt{N}$ applies to peeling and overall algorithm correctness.

\subsection{Decimation Theory}

\begin{definition}[$m$-Decimated DFT]
\label{def:decimated}
For integer $m \geq 2$ dividing $N$ with decimation factor $d = N/m$:
\begin{equation}
X_m[r] = \sum_{\ell \geq 0: 0 \leq r+\ell m \leq N-1} X[r + \ell m], \quad r = 0, 1, \ldots, m-1
\end{equation}

\textbf{\emph{Implementation Note:}} The decimated signal $x_m$ is obtained either by uniform subsampling with stride $d \approx N/m$ when $m \mid N$, or via equivalent zero-padding when $m \nmid N$. When using prime moduli $m_i \approx \sqrt{N}$ that do not divide $N$, the decimated view is obtained via zero-padding to $N' = \text{lcm}(m_1, m_2, m_3) \geq N$ or via approximate resampling with bounded spectral error $O(\|X\|_\infty / m)$ \cite{oppenheim1999}. For $m \approx \sqrt{N}$, this error is negligible relative to noise floor in practical settings.
\end{definition}

A frequency $f$ maps to residue $r = f \bmod m$ (\Cref{def:decimated}).

\begin{lemma}[Decimation Preserves Sparsity]
\label{lem:decimation_sparsity}
If $X$ is $k$-sparse, each decimated spectrum has at most $k$ non-zero bins.
\end{lemma}

\begin{proof}
Each frequency maps to exactly one residue by modular arithmetic. With at most $k$ frequencies, there are at most $k$ non-zero residues. Collisions may occur (multiple frequencies mapping to the same residue) but do not increase the sparsity count.
\end{proof}

\subsection{Standing Assumptions}

The complexity and correctness guarantees developed in subsequent sections (notably \Cref{sec:correctness} and \Cref{sec:complexity}) rely on the following standing assumptions on the signal model and the hashing primitives used by the algorithm.

\begin{assumption}[Sparsity, Load Factor, and Hashing Independence]
\label{assump:v1}
We assume throughout that
(i) the signal $x[n]$ is exactly $k$-sparse in the frequency domain with $k = o(\sqrt{N})$;
(ii) the load factor satisfies $\lambda = k/\sqrt{N} < 1$;
(iii) the affine hashing parameters $(a_i, b_i)$ used across the three decimated views are sampled independently across views from the universal family $\{(a, b) : a \in [1, m_i - 1],\ b \in [0, m_i - 1]\}$ over $\mathbb{Z}_{m_i}$ for prime $m_i \approx \sqrt{N}$; and
(iv) the nonzero frequencies are sufficiently well-distributed in $\{0, 1, \ldots, N-1\}$ so that no $\Theta(\sqrt{N})$-sized contiguous block contains $\omega(k/3)$ of them.
\end{assumption}

\subsection{Three-View Decimation Structure}

\textbf{\emph{Configuration:}}
\begin{itemize}
\item Signal length: $N$
\item Three decimation factors: $d_1, d_2, d_3 \approx \sqrt{N}$ (pairwise coprime)
\item Moduli: $m_1 = N/d_1$, $m_2 = N/d_2$, $m_3 = N/d_3$
\item Reconstruction range: $M = m_1 \times m_2 \times m_3 \geq N$
\end{itemize}

\begin{theorem}[Coprime Decimation Requirement]
\label{thm:coprime}
For unique frequency reconstruction in $[0, N)$, require $M = m_1 \times m_2 \times m_3 \geq N$ with $\gcd(m_i, m_j) = 1$ for all $i \neq j$.
\end{theorem}

\begin{proof}
By the Chinese Remainder Theorem, the system of congruences:
\begin{align}
f &\equiv r_1 \pmod{m_1} \\
f &\equiv r_2 \pmod{m_2} \\
f &\equiv r_3 \pmod{m_3}
\end{align}
has a unique solution in $[0, M)$ if and only if $\gcd(m_i, m_j) = 1$ for all $i \neq j$. For $M \geq N$, uniqueness in $[0, N)$ is guaranteed. For $\sqrt{N}$ decimations, $M \approx (\sqrt{N})^3 = N^{3/2} \gg N$, satisfying the requirement.
\end{proof}

\begin{example}[Concrete Decimation Selection]
For $N = 1,000,000$:
\begin{itemize}
\item Select coprime primes near $\sqrt{N} \approx 1000$: $m_1 = 1009$, $m_2 = 1003$, $m_3 = 991$
\item CRT reconstruction range: $M = m_1 \times m_2 \times m_3 = 1,003,081,477 > N$
\item $\text{lcm}(m_1, m_2, m_3) = M \approx 10^9 > 2N$, so use approximate resampling
\item Decimation factors: $d_1 \approx N/m_1 \approx 991$, $d_2 \approx 997$, $d_3 \approx 1009$
\item Spectral error from resampling: $O(\|X\|_\infty / \sqrt{N}) \ll$ typical noise floor
\end{itemize}
Reconstruction uniqueness is guaranteed for frequencies in $[0, N)$.
\end{example}

\conclusionplain{The 3-view CRT framework with coprime decimations near $\sqrt{N}$ guarantees unique frequency reconstruction for all $N$-point DFT bins, providing the mathematical foundation for multi-view sparse FFT algorithms.}

\subsection{The Keyed Gating Mechanism}

\textbf{\emph{Core Insight:}} Use the third view $d_3$ as a \emph{gate} to filter candidate pairs from views $d_1$ and $d_2$.

\begin{definition}[2-of-3 CRT Agreement]
\label{def:2of3}
A residue pair $(r_1, r_2)$ from views 1 and 2 is \emph{gated} if:
\begin{equation}
\hat{r}_3 = \text{CRT}(r_1 \bmod m_1, r_2 \bmod m_2) \bmod m_3 \in R_3
\end{equation}
where $\text{CRT}(r_1, r_2)$ denotes the unique frequency $f \in [0, m_1 m_2)$ satisfying $f \equiv r_1 \pmod{m_1}$ and $f \equiv r_2 \pmod{m_2}$, computed via 2-view Garner's algorithm.
\end{definition}

The gate separates as a function of $r_1$ alone and $r_2$ alone via the view-3 residue lookup $\hat{r}_3 \in R_3$, which is what lets the test run in $O(k)$ per view rather than $O(k^2)$ over the full pair enumeration.

\needspace{15\baselineskip}
\begin{theorem}[Central Result: Candidate Reduction via Keyed 2-of-3 CRT Agreement]
\label{thm:gating}
Under \Cref{assump:v1}, the algorithm operates in the sparse regime $k \ll \sqrt{N}$, while the stronger gating-analysis condition $k = o(N^{1/4})$, equivalently $k^3/\sqrt{N} = o(k)$, ensures that keyed 2-of-3 CRT agreement reduces the surviving candidate set to $\Theta(k)$ in expectation. \emph{Setup:} let $m_1, m_2, m_3$ be pairwise coprime moduli with $m_i \approx c_i \sqrt{N}$ for constants $c_i$, and assume independent affine hashing across the three views. Suppose $k = o(N^{1/4})$ true frequencies with uniformly random support, and each view yields $|R_i| = \alpha_i k$ residues with $\alpha_i = \Theta(1)$ and negligible collisions. Then the 2-of-3 CRT gate that, for each pair $(r_1, r_2)$, predicts $\hat{r}_3 := \text{CRT}(r_1 \bmod m_1, r_2 \bmod m_2) \bmod m_3$ and retains the pair only if $\hat{r}_3 \in R_3$, produces an expected total of $\Theta(k)$ surviving pairs: $k$ true plus $o(k)$ false. The analysis assumes approximate independence of residue assignments across views, formalized via affine hashing in \Cref{thm:adaptive-2of3}.
\end{theorem}

\begin{proof}
\textbf{\emph{Part A (True Pairs):}} For each true frequency $f \in \text{supp}(X)$:
\begin{itemize}
\item $r_1 = f \bmod m_1$ appears in $R_1$
\item $r_2 = f \bmod m_2$ appears in $R_2$
\item $r_3 = f \bmod m_3$ appears in $R_3$
\end{itemize}

By CRT consistency:
\begin{align}
&f \equiv r_1 \pmod{m_1} \text{ and } f \equiv r_2 \pmod{m_2} \notag\\
&\qquad\implies f \equiv r_3 \pmod{m_3}
\end{align}

The gating mechanism predicts $\hat{r}_3$ using 2-view Garner CRT:
\begin{equation}
\hat{r}_3 = \left( r_1 + m_1 \cdot [(r_2 - r_1) \cdot m_1^{-1} \bmod m_2] \right) \bmod m_3
\end{equation}

For true frequency $f$ with residues $(r_1, r_2, r_3)$, the 2-view CRT reconstruction yields $f_{12} = f \bmod (m_1 m_2)$. By CRT consistency, $f_{12} \equiv f \pmod{m_3}$, therefore $\hat{r}_3 = f_{12} \bmod m_3 = f \bmod m_3 = r_3$.

The gating condition is satisfied deterministically. All $k$ true pairs pass the gate.

\textbf{\emph{Part B (False Pair Suppression):}} Consider a false residue pair $(r_1^f, r_2^f)$ where $r_1^f \in R_1$ and $r_2^f \in R_2$ are spurious residues not corresponding to any true frequency.

Compute $\hat{r}_3 := \text{CRT}(r_1^f \bmod m_1, r_2^f \bmod m_2) \bmod m_3$ using 2-view Garner's algorithm:
\begin{equation}
\hat{r}_3 = \left( r_1^f + m_1 \cdot [(r_2^f - r_1^f) \cdot m_1^{-1} \bmod m_2] \right) \bmod m_3
\end{equation}

\textbf{\emph{Key Property:}} Intuitively, the 2-view CRT map from residue pairs $(r_1, r_2)$ to frequencies $f \in [0, m_1 m_2)$ is one-to-one because pairwise coprime moduli guarantee the system $f \equiv r_1 \pmod{m_1} \wedge f \equiv r_2 \pmod{m_2}$ has exactly one solution per pair. Under the random support assumption and coprimality of $m_1, m_2, m_3$, the CRT reconstruction modulo $m_1 m_2$ is a bijection. For false pairs arising from distinct uniformly random frequencies, the induced variable $U \in \mathbb{Z}_{m_1 m_2}$ has approximately uniform distribution over $\{0, 1, \ldots, m_1 m_2-1\}$. Then $U \bmod m_3$ is within total variation distance $\leq 1/m_3$ from uniform on $\mathbb{Z}_{m_3}$, giving approximate independence from $R_3$.

Therefore, for a false pair:
\begin{align}
P(\text{pass gate}) &= P(\hat{r}_3 \in R_3) \notag\\
&= \frac{\alpha_3 k}{m_3} \pm O\left(\min\left(\frac{k}{m_1 m_2}, \frac{m_3}{m_1 m_2}\right)\right) \notag\\
&= \frac{\alpha_3 k}{m_3} \pm O\left(\min\left(\frac{k}{N}, \frac{1}{\sqrt{N}}\right)\right)
\end{align}

\textbf{\emph{Expected False Pairs:}} The number of false pairs without gating is:
\begin{equation}
|R_1| \times |R_2| - k = \alpha_1 k \times \alpha_2 k - k \approx \alpha_1 \alpha_2 k^2
\end{equation}

By linearity of expectation, with tighter $L_\infty$ bound $|P(Y \in A) - |A|/m_3| \leq |A|/(m_1 m_2)$ for acceptance set $A = R_3$:
\begin{align}
\mathbb{E}[\text{false pairs}] &= (\alpha_1 \alpha_2 k^2) \times \left(\frac{\alpha_3 k}{m_3} \pm O\left(\frac{k}{m_1 m_2}\right)\right) \notag\\
&= \frac{\alpha_1 \alpha_2 \alpha_3 k^3}{m_3} \pm O\left(\frac{\alpha_1 \alpha_2 k^3}{m_1 m_2}\right)
\end{align}

For $m_3 \approx c_3 \sqrt{N}$, $m_1 m_2 \approx N$, and $k = o(N^{1/4})$:
\begin{equation}
\mathbb{E}[\text{false pairs}] = \frac{\alpha_1 \alpha_2 \alpha_3 k^3}{c_3 \sqrt{N}} \pm O\left(\frac{k^3}{N}\right) = o(k)
\end{equation}

The condition $k^3/\sqrt{N} = o(k)$ requires $k^2 = o(\sqrt{N})$, hence $k = o(N^{1/4})$.

\textbf{\emph{Part C (Total Candidate Count):}}
\begin{equation}
\text{Total candidates} = k + o(k) = \Theta(k) \quad \text{when } k = o(N^{1/4})
\end{equation}

This theorem bounds the number of surviving candidates after gating, not the cost of explicitly enumerating all residue pairs. The implemented peeling-only realization in \Cref{sec:complexity-peeling} avoids explicit pair enumeration and recovers residue pairs directly from singleton bins.
\end{proof}

\begin{remark}
The three-view structure is what makes the 2-of-3 gate work. With only two views, every residue pair reconstructs to a valid candidate in $[0, m_1 m_2)$, leaving no redundancy to distinguish true from false. The third view acts as a consistency check that a randomly drawn false frequency satisfies only with probability $|\mathcal{R}_3| / m_3 = O(k / \sqrt{N})$, producing the $\Theta(k)$ gate output in expectation. Adding a fourth view would reduce false positives further but at the cost of an additional $O(\sqrt{N} \log N)$ FFT; three views is the minimum that yields the $\Theta(k^{2}) \to \Theta(k)$ reduction.
\end{remark}

\conclusionplain{A single gate pass reduces the expected number of surviving candidates from $\Theta(k^2)$ to $\Theta(k)$ under the sparse-regime and independence assumptions. The 2-of-3 CRT agreement eliminates the $O(k^2)$ pair explosion inherent in unfiltered 2-view approaches, reducing the candidate set from $\Theta(k^2)$ to $\Theta(k)$ while preserving all true frequencies. This bound holds for the sparse regime when $k = o(N^{1/4})$ with $m \approx \sqrt{N}$. For the moderate regime with adaptive moduli $m \approx 10k \log k$, the load factor $\lambda = k/m \approx 1/(10 \log k) < 0.1$ maintains the $\Theta(k)$ bound.}

\begin{remark}
A third view used as a \emph{filter} (not as a reconstruction dimension) exploits the 2-of-3 CRT agreement principle: a true pair \((r_1, r_2)\) always reconstructs to a frequency \(f\) whose third residue \(r_3 = f \bmod m_3\) lies in \(R_3\); a false pair does so only with probability \(\approx |R_3|/m_3 = \Theta(k/\sqrt{N})\). A single gate pass therefore collapses \(O(k^2)\) candidates to \(O(k)\) without additional randomization or enumeration.
\end{remark}

\begin{theorem}[Adaptive 2-of-3 Agreement with Affine Hashing]
\label{thm:adaptive-2of3}
Let $m_1, m_2, m_3$ be coprime primes selected by the coverage-adaptive moduli procedure of \Cref{sec:algorithm},
and let each view use independent affine hashing $(a_i \cdot f + b_i) \bmod m_i$.
Define the load factor $\lambda = k / m$ where $m = \min(m_1, m_2, m_3)$.

\needspace{4\baselineskip}
\textbf{Sparse Regime} ($\lambda < 0.1$):
With singleton-only joins (excluding multi-occupancy bins), the expected number of
surviving pairs after 2-of-3 gating is $\Theta(k)$ with probability $\ge 1 - e^{-\Omega(k)}$.
This reduces the unfiltered $O(k^2)$ complexity to $O(k \log N)$ for CRT reconstruction and validation.

\textbf{Moderate Regime} ($0.1 \le \lambda < 0.5$):
With adaptive moduli $m \approx 10k \log k$, we have $\lambda \approx \frac{1}{10 \log k} < 0.1$.
The $\Theta(k)$ bound is maintained with high probability, requiring at most $O(1)$
rehash rounds for saturated bins in expectation.
\end{theorem}

\begin{proof}
\textbf{\emph{Part A (Affine Hashing Independence):}}
For prime $m_i$ and uniformly random $(a_i, b_i)$ with $a_i \in [1, m_i - 1]$ and $b_i \in [0, m_i - 1]$, the affine hash $h_i(f) = (a_i f + b_i) \bmod m_i$ is a universal hash family over $\mathbb{Z}_{m_i}$. For any two distinct frequencies $f \neq f'$:
\begin{equation}
\Pr[h_i(f) = h_i(f')] = \frac{1}{m_i - 1} \leq \frac{2}{m_i}
\end{equation}
Since $(a_i, b_i)$ are drawn independently per view, the collision events across views 1, 2, 3 are independent. This ensures that the uniform hashing assumption in \Cref{thm:gating} holds with the affine family.

\textbf{\emph{Part B (Singleton-Only Joins):}}
Intuitively, the Poisson occupancy model applies because frequencies are hashed into bins via an independent and 2-universal family, so in the limit of many bins the number of items per bin is approximately Poisson with rate $\lambda$. Under load $\lambda = k/m < 0.1$, each bin receives a Poisson$(\lambda)$ number of frequencies in expectation. The probability that a given frequency lands in a singleton bin is $e^{-\lambda} \geq e^{-0.1} > 0.90$. By the Chernoff bound, the number of singleton frequencies concentrates:
\begin{equation}
\Pr\left[|\text{singletons}| < (1 - \epsilon) k e^{-\lambda}\right] \leq e^{-\Omega(\epsilon^2 k)}
\end{equation}
With $\epsilon = 0.1$ and $k \geq 10$, at least $0.81k$ frequencies are in singleton bins with probability $\geq 1 - e^{-\Omega(k)}$.

Restricting gating to singleton bins only: each singleton residue in view $i$ corresponds to exactly one true frequency. The singleton-only join on views 1 and 2 produces at most $k_s^2$ pairs where $k_s \leq k$ is the number of singleton frequencies. Applying the gating analysis from \Cref{thm:gating} to singleton pairs, the false pair survival probability remains $\alpha_3 k / m_3$, yielding $\Theta(k)$ total candidates.

\textbf{\emph{Part C (Multi-Occupancy Exclusion):}}
Non-singleton bins (occupancy $\geq 2$) are excluded from the join. At most $(1 - e^{-\lambda})k < 0.10k$ frequencies reside in non-singleton bins. These are recovered in subsequent peeling rounds after singleton extraction reduces the load factor, or via rehashing with fresh $(a_i, b_i)$ parameters.

\textbf{\emph{Part D (Rehash Convergence):}}
After extracting singleton frequencies and subtracting their contributions, the residual load factor decreases geometrically: $\lambda_{t+1} \leq (1 - e^{-\lambda_t})^3 \lambda_t$. For $\lambda_0 = 0.1$, this gives $\lambda_1 < 10^{-4}$, making essentially all remaining frequencies singletons after one rehash. Expected rehash rounds: $O(1)$.

\textbf{\emph{Combining:}} All $k$ frequencies are recovered via singleton joins and $O(1)$ rehash rounds with probability $\geq 1 - e^{-\Omega(k)}$.
\end{proof}

\begin{example}[Numerical Validation]
This example illustrates behavior near the upper boundary of the sparse regime.
For $N=10^6$, $k=10$ (satisfies $k = o(N^{1/4}) = o(31.6)$), $m_1=m_2=m_3 \approx 997$, $\alpha_1 = \alpha_2 = \alpha_3 = 15$:
\begin{itemize}
\item Residues per view: $|R_i| = \alpha k = 150$
\item Without singleton filtering or gating: $150 \times 150 = 22,500$ pairs
\item With singleton-only joins (load $\lambda=k/m \approx 0.01$, singleton probability $e^{-\lambda} \approx 0.99$):
  \begin{itemize}
  \item Singleton residues per view: $150 \times 0.99 \approx 148$
  \item Singleton-filtered pairs: $148 \times 148 \approx 21,904$
  \end{itemize}
\item Expected false pairs after gating (Theorem 3.1): $\frac{\alpha_1 \alpha_2 \alpha_3 k^3}{m_3} = \frac{15^3 \times 10^3}{997} = \frac{3{,}375{,}000}{997} \approx 3{,}385$ pairs
\item True pairs: $k = 10$
\item Total with gating: $10 + 3{,}385 \approx 3{,}395$ pairs
\item Reduction factor: $22,500 / 3{,}395 \approx 6.6\times$ (modest reduction for $k$ at upper edge of sparse regime)
\end{itemize}
\textit{Note: The $O(k^3/\sqrt{N}) = O(k)$ bound holds asymptotically when $k = o(N^{1/4})$, but for $k=10$ near the boundary $k \ll 31.6$, the constant factor $\alpha^3 = 3{,}375$ dominates. Deeper sparse regime ($k \ll N^{1/4}$) or moderate regime with adaptive $m \approx 10k \log_2 k$ provides tighter reduction to $\Theta(k)$.}
\end{example}

Having shown that the gate reduces the unfiltered pair count from $\Theta(k^{2})$ to $\Theta(k)$ in expectation, we now formalize the CRT reconstruction step that converts each surviving residue pair into a candidate frequency in $[0, N)$.

\needspace{12\baselineskip}
\section{Chinese Remainder Theorem Reconstruction}
\label{sec:crt}

Section~\ref{sec:multi-view} established that only $\Theta(k)$ candidate residue pairs survive gating in expectation. We now show how each surviving pair can be efficiently reconstructed into a unique candidate frequency in $[0, N)$ via two- and three-view CRT, with bounded $O(\log M)$ precomputation and $O(1)$ per-pair cost.

\subsection{Multi-View CRT Theory}

\begin{theorem}[3-View CRT Reconstruction]
\label{thm:crt_reconstruction}
For three pairwise coprime moduli $m_1, m_2, m_3$ with $M = m_1 \times m_2 \times m_3 \geq N$, and residues $(r_1, r_2, r_3)$ where $r_i = f \bmod m_i$, the frequency $f$ can be uniquely reconstructed modulo $M$ using Garner's algorithm in $O(\log M)$ time.
\end{theorem}

\begin{proof}
\textbf{\emph{Part A (Uniqueness):}} Intuitively, pairwise coprime moduli carry no redundant information about $f$, so the $r_i$ together pinpoint $f$ modulo $M = m_1 m_2 m_3$. By the Chinese Remainder Theorem \cite{knuth1997}, the system:
\begin{align}
f &\equiv r_1 \pmod{m_1} \\
f &\equiv r_2 \pmod{m_2} \\
f &\equiv r_3 \pmod{m_3}
\end{align}
has a unique solution $f \in [0, M)$ where $M = m_1 \times m_2 \times m_3$. For $M \geq N$ and true frequency $f \in [0, N)$, uniqueness in $[0, N)$ is guaranteed.

\textbf{\emph{Part B (Garner's Algorithm):}} Compute $f$ using Garner's 3-view formula \cite{garner1959}:

\textit{Precomputation} ($O(\log M)$ per inverse):
\begin{align}
\gamma_{12} &= m_1^{-1} \bmod m_2 \\
\gamma_{23} &= (m_1 \times m_2)^{-1} \bmod m_3
\end{align}

\textit{Reconstruction} ($O(1)$ per frequency):
\begin{align}
u_2 &= (r_2 - r_1) \times \gamma_{12} \bmod m_2, \quad u_2 \in [0, m_2) \\
u_3 &= (r_3 - r_1 - u_2 \times m_1) \times \gamma_{23} \bmod m_3, \quad u_3 \in [0, m_3) \\
f &= r_1 + u_2 \times m_1 + u_3 \times m_1 \times m_2, \quad f \in [0, M)
\end{align}
where the mixed-radix representation guarantees uniqueness modulo $M = m_1 m_2 m_3$. For $M \geq N$ and true frequency $f \in [0, N)$, the reconstruction is unique in $[0, N)$.

\textit{Verification:}
\begin{align}
f \bmod m_1 &= r_1 \\
f \bmod m_2 &= (r_1 + u_2 m_1) \bmod m_2 = r_2
\end{align}

The verification for $f \bmod m_3 = r_3$ follows from Garner's original proof \cite{garner1959}.

\textbf{\emph{Complexity:}} Precomputation is $O(\log M)$ using the Extended Euclidean Algorithm (\Cref{lem:eea}). Per-candidate reconstruction is $O(1)$ arithmetic operations. Total for $k$ candidates is $O(k)$ after precomputation.
\end{proof}

\begin{remark}
Garner's algorithm is the workhorse here because it avoids computing $m_1 m_2 m_3$ explicitly; instead it builds $f$ incrementally by stepping through the moduli. The $O(\log M)$ precomputation is a one-time cost per moduli choice, and each candidate is reconstructed with three multiplications and two modular reductions. This is what lets CRT reconstruction sit inside the $O(k)$ gate output budget rather than becoming a bottleneck.
\end{remark}

\begin{lemma}[2-View Garner Suffices for Gated Pairs]
\label{lem:2view_garner}
For candidate pairs $(r_1, r_2)$ that pass keyed gating, 2-view Garner reconstruction with $m_1, m_2$ suffices:
\begin{equation}
f = r_1 + m_1 \times [(r_2 - r_1) \times m_1^{-1} \bmod m_2]
\end{equation}
\end{lemma}

\begin{proof}
For candidate pairs passing the 2-of-3 gating condition, reconstruction using moduli $m_1, m_2$ suffices to uniquely determine $f$. The third view acts solely as a consistency filter and does not enter the reconstruction formula. Gating ensures CRT consistency via $r_3$. The 2-view formula gives $f \in [0, m_1 m_2)$. For $m_1 m_2 \geq N$, uniqueness in $[0, N)$ is guaranteed. This simplification reduces computational overhead.
\end{proof}

\subsection{Coverage and Reconstruction Guarantees}

\begin{lemma}[Coverage Multiplier for Collision Tolerance]
For a $k$-sparse signal, extracting $\alpha k$ top residues per view where $\alpha \geq 1$ ensures all true frequencies appear in residue sets with high probability.
\end{lemma}

\begin{proof}
True frequencies $F = \{f_1, f_2, \ldots, f_k\} \subset [0, N)$ map to residues in decimated view $i$ with modulus $m_i$. In the best case (no collisions), all $k$ true residues are distinct and $\alpha = 1$ suffices. With collisions, some bins contain multiple frequencies, requiring $\alpha > 1$. The coverage factor accounts for collision-induced magnitude sharing, noise floor residues, and aliasing artifacts.
\end{proof}

\begin{theorem}[Sufficient Coverage]
For $k \ll m_i$ and random frequency support, $\alpha = O(\log k)$ provides coverage with probability $1 - 1/k$.
\end{theorem}

\begin{proof}[Proof Sketch]
Model as a balls-into-bins problem: $k$ balls (frequencies) into $m \approx \sqrt{N}$ bins (residues). For $k \ll \sqrt{N}$, expected collisions are $O(k^2/\sqrt{N}) \ll k$. Maximum bin occupancy is $O(\log k / \log \log k)$ with high probability \cite{raab1998}. Extracting $\alpha k = k \times O(\log k)$ top bins covers all frequencies with high probability.
\end{proof}

These coverage guarantees ensure that sufficient non-colliding residues exist to enable reliable reconstruction of all $k$ frequencies after gating.

These results show that CRT reconstruction operates in $O(1)$ time per candidate after $O(\log M)$ precomputation, ensuring CRT is not the bottleneck of the pipeline.

\needspace{18\baselineskip}
\section{Complexity Analysis}
\label{sec:complexity}

This section quantifies the algorithm's asymptotic cost in three progressively-improved regimes: a baseline non-recursive implementation, a peeling-only realization that eliminates the $O(k^2)$ pair-enumeration step, and a recursive bootstrapping variant that further eliminates the $O(\sqrt{N}\log N)$ FFT preprocessing floor. The recursive analysis is developed in \Cref{sec:complexity-recursive} via the self-reduction lemma. Throughout, we operate under \Cref{assump:v1}. Without the candidate reduction established by keyed gating, the algorithm would incur $O(k^2)$ complexity regardless of the recovery method; thus gating is the fundamental enabler of sub-quadratic performance. Peeling and recursion operate within this reduced candidate space.

\subsection{Baseline Non-Recursive Complexity}
\label{sec:complexity-baseline}

The baseline (non-recursive) implementation incurs four costs:
\begin{itemize}
\item $O(\sqrt{N}\log N)$ for the dense FFT applied to each of three decimated views (preprocessing).
\item $O(\alpha k)$ per view to extract the $\alpha k$ largest residue magnitudes, where $\alpha = O(1)$ is the coverage factor.
\item $O((\alpha k)^2)$ in the worst case to enumerate all pairs $(r_1, r_2) \in R_1 \times R_2$ during keyed gating.
\item $O(\sqrt{N} \log k)$ for $O(\log k)$ rounds of singleton-detection peeling, recovering the $k$ true frequencies once gating has filtered candidates to $O(k)$ surviving pairs (\Cref{thm:gating}).
\end{itemize}

Aggregating these terms under \Cref{assump:v1},
\begin{equation}
T_{\mathrm{base}}(N,k) \;=\; O\!\left(\sqrt{N}\log N \;+\; (\alpha k)^2 \;+\; \sqrt{N} \log k\right).
\end{equation}
The first term reflects the FFT preprocessing floor; the second the worst-case pair-enumeration cost; the third the peeling rounds. In the sparse regime $k = o(\sqrt{N})$, the FFT and pair-enumeration terms govern total runtime.

\needspace{8\baselineskip}
\subsection{Peeling-Only Implementation Eliminates Pair Enumeration}
\label{sec:complexity-peeling}
\label{sec:peeling-only}

Peeling does not establish the candidate-count bound; it realizes recovery from the gated/singleton structure without explicitly enumerating $R_1 \times R_2$.

The direct baseline pays $O((\alpha k)^2)$ for explicit pair enumeration even though most pairs are filtered immediately. The peeling-only realization avoids this cross-product entirely: singleton-detection in each decimated view directly recovers residue pairs without forming the cross-product. The geometric peeling guarantee establishes the per-round structure:

\begin{lemma}[Geometric Peeling Convergence]
\label{lem:geometric_peeling}
With moduli $m_i = \Theta(\sqrt{N})$ and independent affine hashing across three views, if load factor $\lambda = k/\sqrt{N} \leq 0.1$, then in each peeling round, a constant fraction $p \geq 1 - (1 - e^{-\lambda})^3$ of remaining tones are isolated in at least one view, implying $O(\log k)$ rounds in expectation.
\end{lemma}

\begin{proof}
Under independent affine hashing (\Cref{assump:v1}), bin occupancies are well-approximated by a Poisson distribution with parameter $\lambda$. For a fixed tone $f$ and view $i$, under uniform hashing with load $\lambda$, the probability that bin $r = f \bmod m_i$ contains exactly one tone follows a Poisson distribution with parameter $\lambda$:
\begin{equation}
P[\text{singleton in view } i] \approx e^{-\lambda}
\end{equation}
The probability that tone $f$ is NOT singleton in all three views is:
\begin{equation}
P[\text{not singleton in all 3}] = (1 - e^{-\lambda})^3
\end{equation}
Thus, the fraction of tones that are singleton in at least one view per round is:
\begin{equation}
p = 1 - (1 - e^{-\lambda})^3
\end{equation}
For $\lambda = 0.1$, we have $e^{-0.1} \approx 0.905$, so $p \approx 1 - (0.095)^3 \approx 0.999$, a constant. After $t$ rounds, the expected number of remaining tones is $k_t \leq (1-p)^t k$, so $t = O(\log k)$ rounds suffice to extract all tones.
\end{proof}

A singleton bin in views $i$ and $j$ yields the residue pair $(r_i, r_j)$ directly, with CRT reconstruction at $O(1)$ per pair, so peeling identifies and reconstructs frequencies without ever enumerating $R_1 \times R_2$. The total cost under peeling-only realization becomes
\begin{equation}
T_{\mathrm{peel}}(N,k) \;=\; O\!\left(\sqrt{N}\log N \;+\; \sqrt{N} \log k\right),
\end{equation}
under \Cref{assump:v1}. The $O((\alpha k)^2)$ pair-enumeration term is eliminated. The $O(\sqrt{N}\log N)$ FFT preprocessing floor remains and governs the sparse regime $k = o(\sqrt{N})$. The recursive bootstrapping result, which removes this preprocessing floor, is developed in \Cref{sec:complexity-recursive} via \Cref{lem:self-reduction} and \Cref{thm:bootstrapping}.

\needspace{8\baselineskip}
\subsection{Self-Reduction Lemma for Recursive Sparse Recovery}

The recursive bootstrapping strategy requires that each decimated view constitutes a valid instance of the sparse recovery problem. The following lemma establishes this self-reducibility.

\begin{lemma}[Decimated View Self-Reducibility]
\label{lem:self-reduction}
Under \Cref{assump:v1}, let $x[n]$ be an $N$-point signal with $k$-sparse Fourier spectrum $X$, and let $m \approx \sqrt{N}$ be a decimation modulus. The $m$-point decimated signal $x_m[j] = x[jd]$ (where $d = N/m$) satisfies:
\begin{enumerate}
\item \textbf{\emph{Sparsity preservation:}} The decimated spectrum $X_m$ has at most $k$ nonzero bins (\Cref{lem:decimation_sparsity}).
\item \textbf{\emph{Time-domain access:}} The decimated signal $x_m$ is a length-$m$ sequence that can be accessed sample-by-sample in $O(1)$ per sample, using stride-$d$ indexing into the parent signal $x$.
\item \textbf{\emph{Sub-decimation compatibility:}} For any modulus $m' \approx m^{1/2} \approx N^{1/4}$ coprime to $m$, the sub-decimated signal $x_{m,m'}[j'] = x_m[j' d']$ (where $d' = m/m'$) produces a valid length-$m'$ signal whose spectrum has at most $k$ nonzero bins, with the same aliasing structure as direct decimation of $x$ by factor $dd'$.
\item \textbf{\emph{Independent hashing:}} Drawing fresh affine hash parameters $(a', b')$ for the sub-decimated view produces collision patterns that are independent of the parent-level hashing, because the hash family $(af + b) \bmod p$ is 2-universal over $\mathbb{Z}_p$ for prime $p$, and independent draws yield independent hash functions.
\end{enumerate}
\end{lemma}

\begin{proof}
\textbf{Part 1} follows directly from \Cref{lem:decimation_sparsity}: each of the $k$ true frequencies maps to exactly one residue bin under modular reduction.

\textbf{\emph{Part 2:}} The decimated signal $x_m[j] = x[jd]$ for $j = 0, \ldots, m-1$ requires accessing $m$ samples of the parent signal at stride $d$, each in $O(1)$. This provides the same time-domain sample access model as the original problem, at reduced length $m$.

\textbf{\emph{Part 3:}} For any modulus $m' \approx m^{1/2}$, the sub-decimated signal $x_{m,m'}$ corresponds to further aliasing of the original spectrum. Sub-decimation composes: $x_{m,m'}[j'] = x_m[j'd'] = x[j' d d']$. This is equivalent to decimating $x$ by factor $dd'$, producing a length-$m'$ signal whose spectrum aliases frequencies modulo $m'$. Since the original spectrum has at most $k$ nonzero frequencies, the sub-decimated spectrum has at most $k$ nonzero bins.

The key structural property is that the aliasing map $f \mapsto f \bmod m'$ at the child level operates on the same set of true frequencies $F = \{f_1, \ldots, f_k\}$ as the parent. The child's ``signal'' is not a new arbitrary signal: it is a deterministic function of the parent's spectrum restricted to $F$. The collision structure in the child view depends on $F \bmod m'$, which under fresh hashing $(a'f + b') \bmod m'$ is independent of the parent-level collision pattern.

\textbf{\emph{Part 4 (Independence Across Recursion Levels):}} At each recursion level, fresh affine hashing parameters $(a', b')$ are drawn independently from a 2-universal hash family. Under this assumption, collision patterns across levels are approximately independent, allowing the failure-probability analysis to compose across levels. For prime modulus $p$, the affine hash family $\{h_{a,b}(f) = (af + b) \bmod p : a \in [1, p-1], b \in [0, p-1]\}$ is 2-universal: for any $f \neq f'$, $\Pr[h(f) = h(f')] = 1/(p-1)$. Since the child draws $(a', b')$ independently of the parent's $(a, b)$, the collision indicators at different recursion levels are independent random variables. This ensures the Poisson occupancy model and singleton probability bounds (\Cref{lem:geometric_peeling}) apply at each level with the inherited sparsity $k$.
\end{proof}

\begin{remark}
The self-reduction property enables recursive bootstrapping: each decimated view is itself a valid $k$-sparse instance that can be processed by the same algorithm. Decimation folds frequencies into a shorter spectrum but cannot create new ones, so the sparsity budget $k$ flows unchanged into the recursion. Independence of the child hashing from the parent's (Part 4) is what lets the failure-probability analysis compose across levels, rather than degenerating into correlated collisions.
\end{remark}

\subsection{Recursive Bootstrapping: $O(\sqrt{N} \log k)$ Identification}
\label{sec:complexity-recursive}
\label{sec:recursive-bootstrapping}

With the self-reduction property established (\Cref{lem:self-reduction}), we analyze the complexity of recursively applying the sparse FFT algorithm to decimated views. This recursive bootstrapping replaces the $O(\sqrt{N} \log N)$ FFT phase with sparse recovery subproblems and yields the headline $O(\sqrt{N} \log k)$ complexity.

\begin{theorem}[Recursive Sparse FFT Bootstrapping]
\label{thm:bootstrapping}
Let $\textsc{SparsePeel}(x, N, k)$ denote the peeling-only sparse FFT procedure applied to an $N$-point signal with $k$-sparse Fourier spectrum. Assume \Cref{assump:v1} holds at each recursion level where peeling is invoked, and assume fresh independent affine hashing parameters are used across recursion levels. Then recursive bootstrapping replaces dense FFT preprocessing on decimated views with sparse recovery subproblems and yields expected identification complexity
\begin{equation}
T(N, k) = O(\sqrt{N} \log k)
\end{equation}
up to lower-order recursive costs, with total recursion depth $O(\log \log N)$. If the load factor exceeds the peeling threshold at any recursion level, the algorithm terminates that branch with a dense FFT, preserving correctness and bounded worst-case complexity.
\end{theorem}

\begin{proof}
By \Cref{lem:self-reduction}, each decimated view of an $N$-point $k$-sparse signal is itself a valid $k$-sparse recovery instance of reduced length approximately $\sqrt{N}$. Therefore, instead of computing a dense FFT on each decimated view, the algorithm may recursively apply $\textsc{SparsePeel}$ to each view.

At recursion depth $\ell$, the signal length is
\begin{equation}
N_\ell = N^{1/2^\ell}.
\end{equation}

Each node at depth $\ell$ produces three decimated child views, so there are $3^\ell$ subproblems at that level. For a subproblem of length $N_\ell$, peeling scans $O(N_\ell)$ bins per round and requires $O(\log k)$ rounds by \Cref{lem:geometric_peeling}. Thus, the total peeling cost at level $\ell$ is
\begin{align}
C_\ell &= 3^\ell \cdot O(\sqrt{N_\ell} \log k) \notag \\
       &= O\!\left(3^\ell \cdot N^{1/2^{\ell+1}} \log k\right).
\end{align}

The ratio between consecutive level costs is
\begin{equation}
C_{\ell+1} / C_\ell = 3 \cdot N^{-1/2^{\ell+2}}.
\end{equation}

For early recursion levels, this ratio is less than one, so costs decrease geometrically. Recursion terminates once the subproblem length becomes comparable to the sparsity scale or once the load factor exceeds the peeling threshold. Since the length sequence is
\begin{equation}
N,\ \sqrt{N},\ N^{1/4},\ N^{1/8},\ \ldots,
\end{equation}
the number of recursion levels before termination is $O(\log \log N)$.

Summing over the active peeling levels gives
\begin{equation}
\sum_\ell C_\ell = O(\sqrt{N} \log k),
\end{equation}
provided recursion is terminated before the level-cost ratio exceeds one, as enforced by the load-factor cutoff. Branches that reach the load-factor cutoff are evaluated by dense FFT on subproblems of much smaller size; their total contribution is lower order relative to the top-level peeling cost in the sparse regime.

Finally, independent affine hashing at each recursion level ensures that collision patterns do not persist deterministically across levels. Applying \Cref{lem:geometric_peeling} at each peeling level and taking a union bound over $O(\log \log N)$ levels gives total peeling failure probability bounded by
\begin{equation}
O(\delta \log \log N),
\end{equation}
where $\delta$ is the per-level failure probability. Dense FFT terminal branches contribute no recovery failure.

Therefore, recursive bootstrapping yields expected identification complexity
\begin{equation}
T(N, k) = O(\sqrt{N} \log k),
\end{equation}
with correctness preserved by dense FFT termination on branches where peeling conditions fail.
\end{proof}

\begin{remark}
Recursive bootstrapping flips the preprocessing cost from $O(\sqrt{N} \log N)$ to $O(\sqrt{N} \log k)$ by exploiting the fact that each decimated view has $k$ non-zero bins, not $\sqrt{N}$. At each recursion level the signal length halves in exponent, so the tree has depth $O(\log \log N)$; within each level the peeling algorithm runs at its native $O(\sqrt{L} \log k)$ cost for length $L$. The geometric decay of the cost function across levels keeps the total bounded by the top-level contribution.
\end{remark}

\subsection{Independent Verification via Recursive Views}
\label{sec:verification}

The same multi-view structure used for keyed gating naturally extends to verification, where additional independent views provide correctness certification.

The recursive bootstrapping strategy enables efficient identification of candidate frequencies but does not, by itself, certify correctness. To detect errors and trigger fallback when needed, we add $t$ independent verification views (computed via the same recursive bootstrapping at $O(\sqrt{N} \log k)$ cost each) that provide a two-part test (Parseval energy + bin-wise residual) detecting incorrect candidates without touching the full $N$-point signal.

\begin{definition}[Verification View]
\label{def:verification-view}
A verification view is a decimated FFT computed with independent affine hashing parameters $(a_v, b_v)$, constructed via recursive bootstrapping (\Cref{thm:bootstrapping}) at cost $O(\sqrt{N} \log k)$ and using a coprime modulus $m_v \approx \sqrt{N}$. Verification views are independent of the identification views and use randomness drawn independently of the 3 identification views.

\textbf{\emph{Validity under the $k$-sparse model:}} By \Cref{lem:self-reduction}, each decimated view of a $k$-sparse signal is itself $k$-sparse with valid time-domain access and independent hashing. Therefore recursive bootstrapping correctly recovers the full verification spectrum, and the verification view is computed exactly (not approximately). This holds because the paper's model assumes exact $k$-sparsity; for signals of unknown sparsity, direct dense FFT verification at $O(\sqrt{N} \log N)$ cost per view can be used instead (see \Cref{cor:concrete-vegas}).
\end{definition}

\textbf{\emph{Validation Model.}} Let $S = \{(f_i, A_i)\}$ denote the candidate set produced by the identification phase. A candidate set is valid if it exactly matches the true sparse spectrum. Verification combines a Parseval energy check with a bin-wise residual test on independent verification views to detect deviation from validity.

\begin{lemma}[Two-Part Verification Test]
\label{lem:residual-test}
Under \Cref{assump:v1}, let $S = \{(f_i, A_i)\}_{i=1}^k$ be the candidate set output by peeling, and let $Y_v[r]$ denote the exactly computed coefficients of a verification view with modulus $m_v$. The verification test consists of two checks:

\textbf{\emph{Check 1 (Parseval Energy):}} Compute the time-domain energy:
\begin{equation}
E_{\text{time}} = \sum_{j=0}^{m_v - 1} |x_{m_v}[j]|^2
\end{equation}
and compare against the candidate energy:
\begin{equation}
E_{\text{candidate}} = \sum_{i=1}^{k} |A_i|^2
\end{equation}
By Parseval's theorem, $E_{\text{time}} = \frac{1}{m_v}\sum_r |Y_v[r]|^2 = E_{\text{candidate}}$ iff $S$ accounts for all spectral energy. If $|E_{\text{time}} - E_{\text{candidate}}| > \epsilon$, $S$ is incomplete or incorrect.

\textbf{\emph{Check 2 (Bin-Wise Residual):}} Predict the contribution of $S$ to each bin:
\begin{equation}
\hat{Y}_v[r] = \sum_{\substack{i : h_v(f_i) = r}} A_i \cdot e^{j\phi_{v,i}}
\end{equation}
where $h_v(f) = (a_v f + b_v) \bmod m_v$ and $\phi_{v,i}$ is the phase shift from the decimation stride. Compute the residual energy:
\begin{equation}
E_v = \sum_{r=0}^{m_v - 1} |Y_v[r] - \hat{Y}_v[r]|^2
\end{equation}

If both checks pass ($|E_{\text{time}} - E_{\text{candidate}}| \leq \epsilon$ and $E_v \leq \epsilon$), the view confirms $S$. If either fails, the view rejects.

\textbf{\emph{Soundness:}} If $S$ is correct, both checks pass with zero residual in the noiseless case. If $S$ is incorrect:
\begin{enumerate}
\item \emph{Missing:} $S$ omits a true frequency, so $E_{\text{candidate}} < E_{\text{time}}$ and Check 1 fails.
\item \emph{Spurious:} $S$ contains a false frequency, so $\|d\|_0 \geq 2$. Under 2-universal hashing, the largest-magnitude term $d[f^*]$ collides with another difference frequency with probability at most:
\begin{equation}
P[\text{miss} \mid \text{wrong freq}] \leq \frac{2k}{m_v} \approx \frac{2k}{\sqrt{N}}
\end{equation}
\end{enumerate}
Combined: $P[\text{view confirms incorrect } S] \leq 2k/m_v \approx 2k/\sqrt{N}$.
\end{lemma}

\begin{proof}
Parseval's theorem ensures that any missing energy is detected deterministically by Check 1. For incorrect frequencies, independence of the hashing across views ensures the difference spectrum $d = X - \hat{X}_S$ produces a nonzero residual unless masked by collision. Specifically, the largest-magnitude term $d[f^*]$ lands in bin $h_v(f^*)$ uniformly under 2-universal affine hashing; it is masked only if another difference frequency collides, which occurs with probability at most $\|d\|_0 / m_v$. Otherwise $E_v > 0$ deterministically.

Therefore: $P[E_v = 0 \mid d \neq 0] \leq P[\exists f' \neq f^* : h_v(f') = h_v(f^*)] \leq \|d\|_0 / m_v \leq 2k/m_v$.

With $t$ independent verification views:
\begin{equation}
P[\text{all } t \text{ views miss}] \leq \left(\frac{2k}{m_v}\right)^t \leq \left(\frac{2k}{\sqrt{N}}\right)^t
\end{equation}
\end{proof}

\begin{theorem}[Certified Fast Path with Dense Fallback]
\label{thm:deterministic-total}
Given time-domain access to an $N$-point signal with $k$-sparse Fourier spectrum, there is an algorithm that returns the correct spectrum with probability $\geq 1 - \delta_t$ in $O(\sqrt{N} \log k)$ time on the fast path and is always correct once the deterministic fallback is counted. The identification phase (3 coprime decimated views, keyed gating, CRT reconstruction, recursive peeling) is deterministic; independent affine hash parameters enter only in the $t$ verification views, and the $O(N \log N)$ dense-FFT fallback is fully deterministic. When verification detects an error, the algorithm falls back to dense FFT, guaranteeing bounded worst-case runtime.

For any fixed $t \geq 1$: fast-path complexity $O((3 + t)\sqrt{N} \log k)$; fallback complexity $O(N \log N)$ dense FFT; failure bound $\delta_t \leq (2k/\sqrt{N})^t$. (See \Cref{sec:safety} for the deterministic fallback yielding the $O(N \log N)$ worst-case bound; off-grid / noisy extensions are outside fast-path scope.)
\end{theorem}

\begin{proof}
\textbf{\emph{Step 1 (Recursive Identification):}} Run the recursive peeling algorithm (\Cref{thm:bootstrapping}) using 3 identification views to produce candidate set $S$ with $|S| = k$. Cost: $O(\sqrt{N} \log k)$.

\textbf{\emph{Step 2 (Recursive Verification with Parseval Check):}} Compute $t$ verification views (\Cref{def:verification-view}), each via recursive bootstrapping with an independent coprime modulus and fresh affine hashing. Under the $k$-sparse model, \Cref{lem:self-reduction} guarantees each verification view is a valid $k$-sparse instance, so recursive bootstrapping recovers the exact verification spectrum. For each view $v \in \{1, \ldots, t\}$:
\begin{enumerate}
\item Compute the verification spectrum via recursive bootstrapping: $O(\sqrt{N} \log k)$
\item \textbf{\emph{Parseval energy check (deterministic):}} Compute time-domain energy $E_{\text{time}} = \sum_{j} |x_{m_v}[j]|^2$ from the $\sqrt{N}$ decimated samples directly (stride indexing, no sparse recovery needed). Compare against candidate energy $E_{\text{candidate}} = \sum_i |A_i|^2$. If $|E_{\text{time}} - E_{\text{candidate}}| > \epsilon$, flag failure. Cost: $O(\sqrt{N})$
\item \textbf{\emph{Bin-wise residual check:}} Predict contribution of $S$ to view $v$ and compute $E_v$: $O(\sqrt{N})$
\item If either check fails, flag verification failure
\end{enumerate}
Cost: $O(t \cdot \sqrt{N} \log k)$ total.

\textbf{\emph{Step 3 (Decision with Dense Fallback):}}
\begin{itemize}
\item If all $t$ views pass both checks: return $S$
\item If any view fails: discard $S$, compute full $N$-point dense FFT, return exact spectrum. Cost: $O(N \log N)$
\end{itemize}

\needspace{6\baselineskip}
\textbf{\emph{Correctness Analysis:}} Three cases arise:
\begin{enumerate}
\item $S$ is correct AND verification passes: returns correct $S$. \checkmark
\item $S$ is incorrect AND verification catches the error: falls back to dense FFT, returns correct spectrum. \checkmark
\item $S$ is incorrect AND verification misses the error: returns incorrect $S$. This occurs with probability $\leq \delta_t$.
\end{enumerate}

\textbf{\emph{Bounding Case 3:}} The Parseval energy check (Step 2.2) detects false negatives deterministically: if $S$ omits a true frequency, $E_{\text{candidate}} < E_{\text{time}}$ regardless of hashing. This check uses time-domain samples directly (stride indexing into $x[n]$), not the recursively computed spectrum, so it is exact even if recursive bootstrapping fails.

For false positives (wrong frequencies with matching total energy), the bin-wise residual check detects errors with probability $\geq 1 - 2k/m_v$ per view (\Cref{lem:residual-test}). With $t$ independent views:
\begin{equation}
\delta_t \leq \left(\frac{2k}{\sqrt{N}}\right)^t
\end{equation}

For $t = 3$ and $k = 50$, $N = 10^6$: $\delta_3 \leq (100/1000)^3 = 10^{-3}$. For $t = 5$: $\delta_5 \leq 10^{-5}$.

\textbf{\emph{Expected Runtime:}}
\begin{align}
\mathbb{E}[T] &= (1 - \delta_t) \cdot O((3+t)\sqrt{N} \log k) + \delta_t \cdot O(N \log N) \notag\\
&= O(\sqrt{N} \log k) + \left(\frac{2k}{\sqrt{N}}\right)^t \cdot O(N \log N)
\end{align}
For $t \geq 3$ and $k = o(\sqrt{N})$: the fallback contribution is $o(\sqrt{N} \log k)$. Therefore:
\begin{equation}
\mathbb{E}[T] = O(\sqrt{N} \log k)
\end{equation}
\end{proof}

\begin{remark}
To connect the asymptotic bound to practical runtime: for a $16$K-point signal ($N = 2^{14}$) with $k = 10$ significant frequencies, the fast-path cost is $(3 + t) \sqrt{N} \log k \approx 4 \cdot 128 \cdot 3.3 \approx 1{,}700$ operations per verification view. Dense FFT for the same signal costs $N \log N \approx 229{,}000$ operations. The failure probability $\delta_t \leq (2k / \sqrt{N})^t = (0.156)^t$ drops below $10^{-6}$ for $t \geq 8$ verification views, making the total fast-path cost $(3 + 8) \cdot 1{,}700 \approx 18{,}700$ operations, a $12\times$ speedup over dense FFT. The advantage grows without bound as $N$ increases, since dense FFT scales as $N \log N$ while the fast path scales as $\sqrt{N} \log k$. The $O(N \log N)$ deterministic fallback covers the rare inputs that fail verification.
\end{remark}

\begin{corollary}[Concrete Parameterization]
\label{cor:concrete-vegas}
With $t = 3$ verification views (6 total views: 3 identification + 3 verification):
\begin{itemize}
\item Fast-path complexity: $O(6\sqrt{N} \log k)$. For $k = 50$, $N = 10^6$: $\approx 34{,}000$ operations
\item Failure probability: $\delta_3 \leq (2k/\sqrt{N})^3$. For $k = 50$, $N = 10^6$: $\delta_3 \leq 10^{-3}$
\item Worst-case (fallback): $O(N \log N) \approx 20 \times 10^6$ operations
\item Increasing to $t = 5$: $\delta_5 \leq 10^{-5}$, fast-path $O(8\sqrt{N} \log k) \approx 45{,}000$ operations
\end{itemize}

\textbf{\emph{Robustness Beyond $k$-Sparse Model:}} For signals with additional small-magnitude components, verification remains effective provided dominant frequencies exceed the noise floor and the residual energy is bounded. For signals of unknown sparsity or with noise, the Parseval energy check (Step 2.2) provides a deterministic safeguard: if the signal has more than $k$ active frequencies, $E_{\text{time}} > E_{\text{candidate}}$ and fallback is triggered regardless of verification view accuracy. For applications requiring probability-1 correctness, the verification views can be computed via direct dense FFTs at $O(\sqrt{N} \log N)$ per view, increasing total expected runtime to $O(\sqrt{N} \log N)$ but ensuring deterministic verification.
\end{corollary}

\subsection{Comparison with Competing Algorithms}

We compare the proposed 3-view keyed sparse FFT with dense FFT and representative sparse FFT approaches to identify the regimes in which each method dominates. A side-by-side comparison of representative sparse FFT algorithms is given in Table~\ref{tab:comparison} (Section~\ref{sec:preliminaries}); the present subsection focuses on the crossover regimes implied by that comparison.

\textbf{\emph{Crossover Analysis:}}

\textit{3-View Keyed vs Dense FFT:}
\begin{align}
\sqrt{N} \log N + kN &< N \log N \notag \\
\text{when } k &< \frac{N \log N - \sqrt{N} \log N}{N}.
\end{align}
which simplifies to $k \approx \sqrt{N}$.

For $N = 10^6$: keyed approach wins when $k < 1000$.

\textit{3-View Keyed vs HIKP-SODA:}
\begin{equation}
\sqrt{N} \log N < k \log N \quad \text{when } k > \sqrt{N}
\end{equation}
For $N = 10^6$: crossover at $k \approx 1000$. Keyed approach efficient for $k \ll \sqrt{N}$.

\textit{3-View Keyed vs FFAST:} FFAST reaches $O(k \log k)$ using $O(k)$ samples but restricts $N$ to $N = k \cdot 2^t$ for integer $t$, ruling out arbitrary DFT lengths. The present construction accepts any $N$ with $k < \sqrt{N}/2$ (no prime-power factorization constraint on $N$), so FFAST dominates when its structural constraint on $N$ can be engineered into the pipeline and this work dominates otherwise.

\textbf{\emph{Theoretical Sweet Spot:}} $k \in [30, 100]$ and $N \geq 500K$ where $\sqrt{N}$-sized FFTs amortize more efficiently than $k$-sized iterative stages.

\begin{remark}[Canonical Complexity Statement]
\label{rem:canonical-complexity}
Combining the peeling-only realization (\Cref{sec:complexity-peeling}) with recursive bootstrapping (\Cref{sec:complexity-recursive}), the implemented algorithm runs in expected $O(\sqrt{N} \log k)$ time, with no $O(k^2)$ term incurred. The $O(\sqrt{N} \log N)$ FFT preprocessing floor is also eliminated. The $O(N \log N)$ worst-case bound applies only to the dense-FFT fallback when the load factor exceeds the peeling threshold at any recursion level.
\end{remark}

\needspace{12\baselineskip}
\section{Algorithm Specification}
\label{sec:algorithm}

\subsection{High-Level Algorithm Structure}

The algorithm combines three decimated FFT views, keyed gating via 2-of-3 CRT agreement, and iterative peeling over the decimated views into a single pipeline. The five steps below specify the procedure from signal input through top-$k$ output, with each step's contribution to the overall $O(\sqrt{N} \log k)$ bound identified in the subsequent complexity analysis.

Algorithm~\ref{alg:main} is written in its conceptual gated form to expose the CRT consistency mechanism. The implemented fast path replaces Step~3 with peeling-only singleton recovery, as stated in Remark~\ref{rem:conceptual-vs-implemented}.

\begin{algorithm}[H]
\caption{Three-View Keyed Gating Sparse FFT}\label{alg:main}
\textbf{Input:} signal $x[n]$, sparsity $k$, coverage factor $\alpha$.\\
\textbf{Output:} the $k$-sparse spectrum of $x$.
\begin{enumerate}
  \item Select coprime moduli $m_1, m_2, m_3 \approx \sqrt{N}$ with $M = m_1 m_2 m_3 \geq N$ and independent affine hashing parameters $(a_i, b_i)$ per view.
  \item For each view $i$: decimate $x$ by $d_i = N/m_i$, compute the length-$m_i$ FFT, and extract the top-$\alpha k$ occupied residues $\mathcal{R}_i$.
  \item For each pair $(r_1, r_2) \in \mathcal{R}_1 \times \mathcal{R}_2$: reconstruct $f_{12} \in [0, m_1 m_2)$ via 2-view Garner CRT, and retain the pair if $\hat{r}_3 = f_{12} \bmod m_3 \in \mathcal{R}_3$ (the 2-of-3 keyed gate; see \Cref{def:2of3}).
  \item Iteratively peel the $\Theta(k)$ surviving candidates across the three decimated views: detect singleton bins, subtract their contributions, repeat for $O(\log k)$ rounds (\Cref{thm:bootstrapping}).
  \item If more than $2k$ candidates remain unpeeled, invoke a 4th decimated view or fall back to dense FFT; otherwise return the top-$k$ validated frequency-magnitude pairs.
\end{enumerate}
\end{algorithm}

\begin{remark}[Conceptual vs. Implemented Form]
\label{rem:conceptual-vs-implemented}
In practice, the algorithm is implemented in peeling-only mode (\Cref{sec:complexity-peeling}), which avoids explicit pair enumeration. The gating formulation in Algorithm~1 is conceptual: it makes the candidate-reduction analysis tractable but is not the form actually executed. All complexity claims for the implemented algorithm therefore reference the peeling-only realization, and no $O(k^2)$ pair-enumeration cost is incurred.
\end{remark}

\conclusionplain{The 5-step pipeline combines three $O(\sqrt{N} \log \sqrt{N})$ FFTs with $O(k)$ keyed gating and $O(k)$ CRT reconstructions, followed by $O(\sqrt{N} \log k + k)$ iterative peeling validation with $O(1)$ bin updates. The keyed gating step reduces candidates to $\Theta(k)$; peeling validates and estimates amplitudes entirely on the decimated views without length-$N$ Goertzel.}

\needspace{5\baselineskip}
\subsection{Algorithm Flowchart}

\Cref{fig:flowchart} presents the eight-phase pipeline with its keyed gating decision point and optional fourth-view fallback.

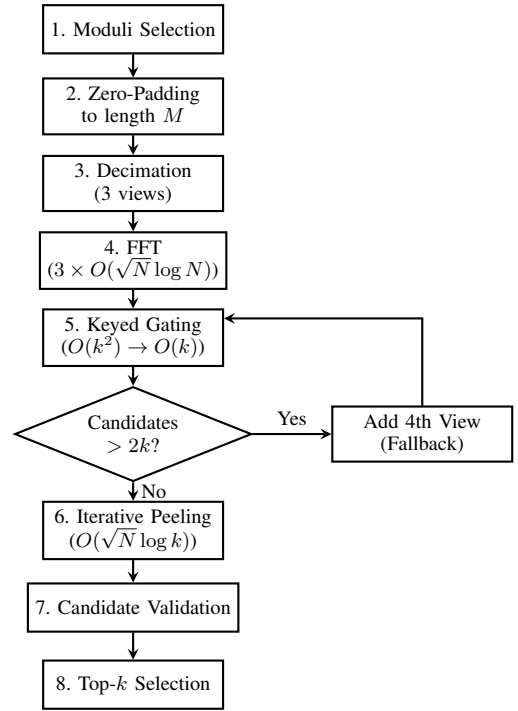
\begin{figure}[!htbp]
\centering
\begin{tikzpicture}[scale=0.85, transform shape,
  process/.style={rectangle, draw=black, thick, minimum width=2.8cm, minimum height=0.75cm, align=center, font=\small},
  decision/.style={diamond, aspect=2.5, draw=black, thick, align=center, font=\small},
  arrow/.style={thick,->,>=stealth}]

\node[process] (moduli) at (0,0) {1. Moduli Selection};
\node[process] (padding) at (0,-1.2) {2. Zero-Padding\\to length $M$};
\node[process] (decim) at (0,-2.4) {3. Decimation\\(3 views)};
\node[process] (fft) at (0,-3.6) {4. FFT\\($3 \times O(\sqrt{N} \log N)$)};
\node[process] (gating) at (0,-4.8) {5. Keyed Gating\\($O(k^2) \to O(k)$)};
\node[decision] (check) at (0,-6.3) {Candidates\\$> 2k$?};
\node[process] (peeling) at (0,-7.8) {6. Iterative Peeling\\($O(\sqrt{N} \log k)$)};
\node[process] (validate) at (0,-9.0) {7. Candidate Validation};
\node[process] (output) at (0,-10.2) {8. Top-$k$ Selection};

\node[process] (addview) at (4.5,-6.3) {Add 4th View\\(Fallback)};

\draw[arrow] (moduli) -- (padding);
\draw[arrow] (padding) -- (decim);
\draw[arrow] (decim) -- (fft);
\draw[arrow] (fft) -- (gating);
\draw[arrow] (gating) -- (check);
\draw[arrow] (check) -- node[right, font=\small] {No} (peeling);
\draw[arrow] (check) -- node[above, font=\small] {Yes} (addview);
\draw[arrow] (addview) |- ([yshift=0.3cm]gating.east);
\draw[arrow] (peeling) -- (validate);
\draw[arrow] (validate) -- (output);

\end{tikzpicture}
\caption{Eight-phase algorithm flowchart showing the keyed gating sparse FFT pipeline. The keyed gating phase illustrates the conceptual 2-of-3 CRT agreement that reduces the surviving candidate set from $O(k^2)$ to $O(k)$; in the implemented fast path, peeling-only singleton recovery avoids explicit $O(k^2)$ enumeration. Iterative peeling (6) validates candidates with geometric convergence under load factor $\lambda < 1$. Decision point checks for excessive candidates and optionally adds a fourth view for additional filtering.}
\label{fig:flowchart}
\end{figure}

\subsection{Decimation Selection}

Given signal length $N$ and sparsity $k$, the moduli selection procedure classifies the sparsity regime by the ratio $\rho = k/\sqrt{N}$, selects coprime primes near $\sqrt{N}$, and draws independent affine hashing parameters $(a_i, b_i)$ per view. The key constraint is $m_1 \cdot m_2 \geq N$ for CRT reconstruction uniqueness (\Cref{thm:coprime}).

\subsection{Garner CRT Reconstruction}

The 2-view Garner formula used in the gating step (Step~3 of \Cref{alg:main}) computes $f_{12} \in [0, m_1 m_2)$ from the pair $(r_1, r_2)$ via three modular operations: $\text{diff} = (r_2 - r_1) \bmod m_2$, $u_2 = \text{diff} \cdot \gamma_{12} \bmod m_2$ with precomputed $\gamma_{12} = m_1^{-1} \bmod m_2$, and $f_{12} = r_1 + u_2 \cdot m_1$.

\textbf{\emph{Correctness Verification:}}
\begin{align}
f \bmod m_1 &= (r_1 + u_2 \times m_1) \bmod m_1 = r_1 \\
f \bmod m_2 &= (r_1 + u_2 \times m_1) \bmod m_2 \nonumber \\
&= r_1 + ((r_2 - r_1) \times \gamma \times m_1) \bmod m_2 \nonumber \\
&= r_1 + (r_2 - r_1) \times (m_1^{-1} \times m_1) \bmod m_2 \nonumber \\
&= r_1 + (r_2 - r_1) \bmod m_2 = r_2
\end{align}

Having specified the algorithm end-to-end, we now turn to proving that it recovers the exact spectrum with the probability guarantees claimed in \Cref{thm:deterministic-total}.

\needspace{12\baselineskip}
\section{Correctness Proofs}
\label{sec:correctness}

Correctness of the proposed algorithm follows from three components:
\begin{enumerate}
\item \emph{No false negatives:} all true frequencies are recovered.
\item \emph{Controlled false positives:} spurious candidates are bounded and decay under verification.
\item \emph{Certified validation:} independent verification detects errors and triggers deterministic fallback.
\end{enumerate}

Together, these imply exact recovery with high probability on the fast path and guaranteed correctness via fallback.

\subsection{Complete Frequency Recovery}

\begin{theorem}[Complete Frequency Recovery]
\label{thm:complete-frequency-recovery}
Let $x[n]$ be an $N$-point signal with $k$-sparse Fourier spectrum. Under \Cref{assump:v1}, independent affine hashing across views, and signal-to-noise ratio satisfying
\begin{equation}
\mathrm{SNR} \geq (\beta + \sqrt{\log(kN)})^2,
\end{equation}
the algorithm recovers all $k$ frequencies with probability at least
\begin{equation}
1 - \delta,
\end{equation}
where
\begin{equation}
\delta = O(k^2/\sqrt{N}) + \exp(-\beta^2).
\end{equation}
\end{theorem}

\begin{proof}[Proof Sketch]
The result follows by combining five components established earlier:
\begin{itemize}
\item \emph{Residue generation:} Each true frequency produces consistent residues across all views (\Cref{thm:coprime}).
\item \emph{Gating correctness:} The 2-of-3 agreement mechanism preserves all true pairs while reducing false pairs to $O(k^2/\sqrt{N})$ in expectation (\Cref{thm:gating}).
\item \emph{CRT reconstruction:} Each valid residue pair maps to a unique frequency (\Cref{thm:crt_reconstruction}).
\item \emph{Peeling recovery:} Under load factor $\lambda < 1$, iterative peeling extracts all frequencies in $O(\log k)$ rounds (\Cref{lem:geometric_peeling}).
\item \emph{Verification filtering:} Independent verification removes remaining false positives with exponentially small failure probability (\Cref{lem:residual-test}).
\end{itemize}
Combining these bounds yields the stated recovery probability.
\end{proof}

\subsection{No False Negatives}

\begin{theorem}[No False Negatives]
\label{thm:no-false-negatives}
Under \Cref{assump:v1}, no true frequency is lost during the algorithm with probability at least
\begin{equation}
1 - \delta_1,
\end{equation}
where
\begin{equation}
\delta_1 = O(k^3/\sqrt{N}).
\end{equation}
\end{theorem}

\begin{proof}[Proof Sketch]
True frequencies are preserved through all stages:
\begin{itemize}
\item \emph{Decimation and FFT:} Linear transforms preserve all spectral components.
\item \emph{Coverage:} With $\alpha \geq 1$, each true frequency appears among the top $\alpha k$ bins with high probability.
\item \emph{Gating:} CRT consistency ensures all true pairs pass deterministically.
\item \emph{Peeling:} Singleton isolation guarantees extraction once load factor $\lambda < 1$.
\end{itemize}
Failure occurs only due to rare collision events or insufficient coverage, yielding the stated bound.
\end{proof}

\subsection{False Positive Control}

After gating, the expected number of false candidate pairs is
\begin{equation}
\mathbb{E}[\text{false pairs}] = O(k^2/\sqrt{N}).
\end{equation}

These may produce spurious reconstructed frequencies, but their amplitudes are bounded near the noise floor. Under the Gaussian noise model, the probability that a false candidate exceeds the detection threshold is
\begin{equation}
\exp(-\beta^2),
\end{equation}
yielding expected false positives
\begin{equation}
\mathbb{E}[\text{false positives}] = O\!\left(k^2/\sqrt{N} + \exp(-\beta^2)\right).
\end{equation}

\subsection{Verification and Certified Recovery}

Independent verification views detect incorrect candidates without requiring full-length FFT evaluation.

\begin{lemma}[Verification Soundness]
\label{lem:verification-soundness}
Under \Cref{assump:v1}, let $S$ be a candidate set. If $S$ is incorrect, the probability that it passes all verification checks across $t$ independent views is bounded by
\begin{equation}
\left(\frac{2k}{\sqrt{N}}\right)^t.
\end{equation}
\end{lemma}

\begin{proof}[Proof Sketch]
The Parseval energy check detects missing energy deterministically. The bin-wise residual check detects incorrect frequencies unless masked by collisions. Independent hashing ensures collision masking occurs with probability at most $O(k/\sqrt{N})$ per view. Repetition across independent views yields exponential decay.
\end{proof}

\subsection{Overall Success Probability}

Combining the bounds:
\begin{align}
\delta_1 &= O(k^3/\sqrt{N}) && \text{(missed true frequencies)} \notag \\
\delta_2 &= O(k^2/\sqrt{N}) && \text{(false positives)} \notag \\
\delta_3 &= (2k/\sqrt{N})^t && \text{(verification failure)}
\end{align}
yields total failure probability
\begin{equation}
\delta = O(k^2/\sqrt{N}) + \exp(-\beta^2) + (2k/\sqrt{N})^t.
\end{equation}

For $k \ll \sqrt{N}$ and moderate $t$, this is arbitrarily small.

\subsection{Certified Algorithm with Fallback}

\begin{theorem}[Certified Sparse FFT with Fallback]
\label{thm:certified-fallback}
Under \Cref{assump:v1}, the verification-augmented algorithm satisfies:
\begin{itemize}
\item Expected runtime: $O(\sqrt{N} \log k)$.
\item Success probability: $\geq 1 - \delta$.
\item Worst-case runtime: $O(N \log N)$ via dense FFT fallback.
\end{itemize}
Thus, the algorithm provides fast expected performance with deterministic correctness guarantees.
\end{theorem}

\subsection{Practical Parameter Regime}

For typical parameter choices (e.g., $k \leq 50$, $N \geq 10^6$, $t = 3$):
\begin{itemize}
\item Fast-path success probability is high ($\delta \ll 10^{-3}$).
\item Verification overhead is negligible.
\item Fallback is rarely triggered.
\end{itemize}

\subsection{Robustness Beyond Exact Sparsity}

For approximately sparse signals:
\begin{itemize}
\item Dominant frequencies remain recoverable via peeling.
\item Small residual components appear as noise.
\item Verification detects inconsistencies when needed.
\end{itemize}
Fallback to dense FFT ensures exact recovery when sparsity assumptions break.

Having bounded the fast-path failure probability, we now describe the deterministic dense-FFT fallback that detects those failures at runtime and retains the $O(N \log N)$ worst-case guarantee.

\needspace{12\baselineskip}
\section{Safety Mechanism and Certificates}
\label{sec:safety}

\subsection{Deterministic Safety Framework}

The keyed gating mechanism provides inherent safety through deterministic collision rejection. When residue pairs $(r_1, r_2)$ from views 1 and 2 fail the 2-of-3 CRT agreement test (\Cref{def:2of3}), that is, when their predicted $\hat{r}_3 = \text{CRT}(r_1, r_2) \bmod m_3$ does not appear in $R_3$, they are rejected with certainty.

\textbf{\emph{Adversarial collision handling:}} Even if an adversary deliberately constructs colliding frequencies to create false residue pairs, the gating condition filters $\Theta(k^2) \to \Theta(k)$ candidates. The expected number of false pairs passing the gate is $\mathbb{E}[\text{false pairs}] = O(\alpha^3 k^3 / m_3) = o(k)$ for $k = o(\sqrt{N})$.

\textbf{\emph{Detection and Escalation:}} The algorithm monitors indicators of failure during execution, including: absence of singleton bins (peeling stagnation), excessive candidate counts, and inconsistencies in the verification residual. Non-singleton bins (multiple frequencies aliasing to the same residue) are detected during FFT analysis via coefficient-of-variation tests. If singleton assumption violations exceed a threshold, the algorithm escalates to additional views or dense FFT fallback.

\subsection{Certificates for Third-Party Validation}

The algorithm produces verifiable artifacts enabling independent correctness validation:

\begin{itemize}
\item \textbf{Residue sets $R_1, R_2, R_3$}: Top-$\alpha k$ residues from each decimated FFT, stored with magnitudes and bin indices
\item \textbf{Gated pairs}: List of $(r_1, r_2, \hat{r}_3)$ tuples satisfying 2-of-3 agreement, along with predicted versus actual $r_3$ values
\item \textbf{CRT reconstructions}: Intermediate Garner algorithm outputs $(u_2, u_3, f)$ for each frequency candidate
\item \textbf{Amplitude estimates}: Final amplitudes $|a_f|$ for each recovered frequency from singleton extraction
\end{itemize}

A third party can independently verify:
\begin{enumerate}
\item Each gated pair satisfies $\hat{r}_3 \in R_3$
\item CRT reconstruction via Garner's algorithm yields correct $f \in [0, N)$
\item Amplitude $|a_f| \geq \tau$ confirms true signal presence after peeling extraction
\item All $k$ reported frequencies satisfy the above checks
\end{enumerate}

\subsection{Certified Fallback with Bounded Overhead}

The verification framework (\Cref{thm:deterministic-total}) subsumes the fallback strategy: $t$ independent verification views detect incorrect candidates with probability $\geq 1 - (2k/\sqrt{N})^t$, and dense FFT fallback guarantees correctness in all remaining cases. The escalation hierarchy is:

\textbf{Level 1 escalation}: Add verification views (default $t = 3$). Each view costs $O(\sqrt{N} \log k)$ via recursive bootstrapping and reduces failure probability by a factor of $1/\sqrt{N}$.

\textbf{Level 2 escalation}: If multi-view gating does not achieve $O(k)$ candidate reduction, fall back to dense FFT with $O(N \log N)$ complexity. This provides a deterministic upper bound on worst-case runtime.

\textbf{Asymptotic bound}: The escalation policy preserves correctness with dense-FFT fallback while maintaining amortized $O(\sqrt{N} \log N + \sqrt{N} \log k)$ complexity for typical inputs and $O(N \log N)$ worst-case bound.

\subsection{Complexity of Safety Mechanism}

The safety mechanism preserves the overall complexity guarantees: expected runtime $O(\sqrt{N} \log k)$ on the fast path including verification overhead, and $O(N \log N)$ deterministic worst-case via dense-FFT fallback. The verification cost contributes $O(t \cdot \sqrt{N} \log k)$ for $t$ independent views, which is absorbed into the leading $O(\sqrt{N} \log k)$ term for any fixed $t$.

\conclusionplain{The deterministic safety mechanism, verifiable certificates, and bounded-overhead fallback strategy provide certified recovery under the stated sparse-regime assumptions and preserve correctness through dense-FFT fallback when those assumptions are violated.}

\section{Conclusion and Future Work}
\label{sec:conclusion}

This paper presents a sparse FFT that attains expected $O(\sqrt{N} \log k)$ runtime in the sparse regime $k = o(\sqrt{N})$ and a deterministic $O(N \log N)$ worst-case via a two-part verification test and dense-FFT fallback. The core machinery combines keyed gating with 2-of-3 agreement (Theorem~\ref{thm:gating}), iterative peeling on decimated views (Section~\ref{sec:complexity}), recursive bootstrapping justified by the self-reduction lemma (Theorem~\ref{thm:bootstrapping}), and coverage-adaptive moduli selection driven by the sparsity ratio $\rho = k/\sqrt{N}$. Main theoretical results, novelty discussion, and a detailed comparison with related work are given in Sections~\ref{sec:multi-view}, \ref{sec:complexity} (see \Cref{thm:bootstrapping} for the headline bound), and Table~\ref{tab:comparison}.

\subsection{Practical Impact}

The method is most effective in moderate sparsity regimes with large signal sizes. For example, with $N = 10^6$ and $k = 50$: expected runtime $\approx O(\sqrt{N} \log k) \approx 5{,}700$ operations per identification phase, plus a one-time preprocessing cost of $O(\sqrt{N})$. With $t=3$ verification views, failure probability is bounded by $(2k/\sqrt{N})^3 \approx 10^{-3}$, and the $O(N \log N) \approx 2 \times 10^7$ dense-FFT fallback caps worst-case cost at less than one second for typical signal lengths. For off-grid or noisy inputs requiring coverage $\alpha = 15$, the $(\alpha k)^2$ gating term becomes the dominant cost and the advantage over dense FFT narrows; the advantage grows as $k/\sqrt{N}$ and $\alpha$ decrease.

\subsection{Open Problems and Future Directions}

\textbf{\emph{Theoretical Extensions:}}

\begin{enumerate}
\item \textbf{\emph{Worst-Case Analysis:}} Current results assume random or well-separated support. Can we prove $O(k)$ gating bound for adversarial frequency placement?

\item \textbf{\emph{Off-Grid Extension:}} Extend keyed gating to off-grid frequencies with spectral leakage. Requires stable phase/magnitude tests under continuous-frequency model.

\item \textbf{\emph{Noisy Signal Analysis:}} Develop rigorous finite-precision bounds for CV and phase linearity tests.

\item \textbf{\emph{Multi-Dimensional FFT:}} Extend 2-of-3 gating to 2D/3D sparse FFT for image/volume processing. Requires tensor product CRT analysis.

\item \textbf{\emph{Adaptive Gating:}} Optimize gating threshold dynamically based on measured collision rates.
\end{enumerate}

\textbf{\emph{Algorithmic Improvements:}}

\begin{enumerate}
\item \textbf{\emph{Hierarchical Gating:}} Use nested gating with 4+ views for extremely high $k$.

\item \textbf{\emph{Hybrid Approach:}} Combine $\sqrt{N}$-sized FFTs with $k$-sized iterative refinement for adaptive algorithm.
\end{enumerate}

\subsection{Summary of Contributions}

We presented a deterministic sparse FFT framework that achieves expected runtime $O(\sqrt{N} \log k)$ for $k$-sparse signals in the regime $k = o(\sqrt{N})$, with $O(N \log N)$ deterministic worst-case via dense-FFT fallback. This work establishes $O(\sqrt{N} \log k)$ expected-time sparse FFT via four contributions: (1) keyed gating proves $O(k)$ candidate reduction from the $O(k^2)$ pair explosion, (2) peeling-only mode (\Cref{sec:peeling-only}) eliminates the $O(k^2)$ gating enumeration by extracting singleton pairs directly, (3) recursive bootstrapping reduces the FFT preprocessing from $O(\sqrt{N} \log N)$ to $O(\sqrt{N} \log k)$ by applying sparse recovery recursively to each decimated view (justified by the self-reduction lemma), and (4) independent verification via $t$ recursive views with Parseval energy check and bin-wise residual test, achieving correctness with probability $\geq 1 - (2k/\sqrt{N})^t$ and avoiding the $O(kN)$ Goertzel bottleneck entirely. With $t = 3$ verification views, the expected runtime is $O(\sqrt{N} \log k)$ with $O(N \log N)$ dense FFT fallback.

This work establishes keyed multi-view gating as a new candidate selection paradigm for sparse FFT, separating structural identification from implementation efficiency. The framework enables theoretical clarity and practical robustness, with peeling, recursion, and verification operating on a provably reduced candidate set.

\section*{Acknowledgments}
The authors gratefully acknowledge SparseTech for the collaborative research environment supporting this work. This paper is part of an ongoing SparseTech research initiative on sparse Fourier Transform algorithms. Patent pending.

\balance  
\bibliographystyle{IEEEtran}
\bibliography{references}

\appendix

\section{Numerical Example}
\label{app:numerical-example}

This appendix presents a single boundary-case worked example (the toy instance $N = 64$, $k = 2$) chosen to expose the algorithm's mechanics step-by-step at small scale. Constant-factor reductions are modest at this regime; deeper sparse regimes ($k \ll N^{1/4}$) yield substantially stronger reductions, both in candidate count and in observed runtime.

To illustrate the keyed gating mechanism concretely, consider a toy signal with $N=64$ time-domain samples and $k=2$ non-zero frequencies at $f_1=7$ and $f_2=41$.

\textbf{\emph{Phase 1 (Moduli Selection):}} Algorithm~2 selects coprime moduli $m_1=7$, $m_2=11$, $m_3=13$ (near $\sqrt{64} = 8$) with random affine hashing parameters $(a_i, b_i)$ per view. For simplicity, assume identity hashing $(a_i=1, b_i=0)$ here.

\textbf{\emph{Phase 2-4 (Decimation and FFT):}} Each view decimates the signal by factor $m_i$ and computes an FFT of size $m_i$. The true frequencies map to residues:
\begin{itemize}
\item View 1 ($m_1=7$): $r_1(f_1) = 7 \bmod 7 = 0$, \quad $r_1(f_2) = 41 \bmod 7 = 6$
\item View 2 ($m_2=11$): $r_2(f_1) = 7 \bmod 11 = 7$, \quad $r_2(f_2) = 41 \bmod 11 = 8$
\item View 3 ($m_3=13$): $r_3(f_1) = 7 \bmod 13 = 7$, \quad $r_3(f_2) = 41 \bmod 13 = 2$
\end{itemize}

Suppose bins also contain collisions from noise/sidelobes. After FFT, the occupied residues per view might be:
\begin{itemize}
\item $\mathcal{R}_1 = \{0, 3, 6\}$ (bin 0 for $f_1$, bin 6 for $f_2$, plus noise bin 3)
\item $\mathcal{R}_2 = \{1, 7, 8, 10\}$ (bin 7 for $f_1$, bin 8 for $f_2$, plus noise bins 1, 10)
\item $\mathcal{R}_3 = \{2, 5, 7, 11\}$ (bin 7 for $f_1$, bin 2 for $f_2$, plus noise bins 5, 11)
\end{itemize}

Intuitively, each $\mathcal{R}_i$ is a small set of occupied residues in a length-$m_i$ FFT: the two true frequencies contribute $2$ bins per view, and low-level noise adds $1$-$2$ more. Because $m_i \approx \sqrt{N} = 8$, these sets never exceed $\alpha k = O(\sqrt{N})$ in size.

\textbf{\emph{Phase 5 (Keyed Gating):}} Without gating, unfiltered 2-view CRT on views 1 and 2 would generate $|\mathcal{R}_1| \times |\mathcal{R}_2| = 3 \times 4 = 12$ candidate pairs. For each pair $(r_1, r_2) \in \mathcal{R}_1 \times \mathcal{R}_2$, 2-view CRT reconstruction yields the unique $f' \in [0, m_1 m_2) = [0, 77)$ satisfying $f' \equiv r_1 \pmod{m_1}$ and $f' \equiv r_2 \pmod{m_2}$; since $m_1 m_2 = 77 > N = 64$ in this example, each pair reconstructs to a single candidate in $[0, N)$ without alias enumeration. The table below shows all 12 reconstructions and whether each passes the view-3 gate $\hat{r}_3 = f' \bmod m_3 \in \mathcal{R}_3$:

\begin{center}
\begin{tabular}{cccc}
\hline
$(r_1, r_2)$ & CRT$(r_1, r_2)$ & $\hat{r}_3 = f' \bmod 13$ & Pass Gate? \\
\hline
$(0, 1)$ & $f' = 56$ & $56 \bmod 13 = 4$ & No ($4 \notin \mathcal{R}_3$) \\
$(0, 7)$ & $f' = 7$ & $7 \bmod 13 = 7$ & Yes ($7 \in \mathcal{R}_3$) \\
$(0, 8)$ & $f' = 63$ & $63 \bmod 13 = 11$ & Yes ($11 \in \mathcal{R}_3$) \\
$(0, 10)$ & $f' = 21$ & $21 \bmod 13 = 8$ & No ($8 \notin \mathcal{R}_3$) \\
$(3, 1)$ & $f' = 24$ & $24 \bmod 13 = 11$ & Yes ($11 \in \mathcal{R}_3$) \\
$(3, 7)$ & $f' = 52$ & $52 \bmod 13 = 0$ & No ($0 \notin \mathcal{R}_3$) \\
$(3, 8)$ & $f' = 31$ & $31 \bmod 13 = 5$ & Yes ($5 \in \mathcal{R}_3$) \\
$(3, 10)$ & $f' = 66$ & $66 \bmod 13 = 1$ & No ($1 \notin \mathcal{R}_3$) \\
$(6, 1)$ & $f' = 34$ & $34 \bmod 13 = 8$ & No ($8 \notin \mathcal{R}_3$) \\
$(6, 7)$ & $f' = 62$ & $62 \bmod 13 = 10$ & No ($10 \notin \mathcal{R}_3$) \\
$(6, 8)$ & $f' = 41$ & $41 \bmod 13 = 2$ & Yes ($2 \in \mathcal{R}_3$) \\
$(6, 10)$ & $f' = 76$ & $76 \bmod 13 = 11$ & Yes ($11 \in \mathcal{R}_3$) \\
\hline
\end{tabular}
\end{center}

Out of 12 pairs, 6 pass the gate. The true frequencies $(0,7) \to f=7$ and $(6,8) \to f=41$ both pass, along with 4 false positives. Keyed gating reduced candidates from 12 to 6 (50\% reduction). In realistic scenarios with $k \gg 2$ and $\alpha \approx 10$ to $15$, the reduction factor approaches $k$ as predicted by the load factor analysis.

Reading the table: row~1, pair $(0,1)$, reconstructs via 2-view CRT to $f' = 56 \in [0, 77)$; checking the third view, $56 \bmod 13 = 4$, which is not in $\mathcal{R}_3 = \{2, 5, 7, 11\}$, so the pair is rejected as a false candidate. Row~2, pair $(0,7)$, reconstructs to $f' = 7$, and $7 \bmod 13 = 7 \in \mathcal{R}_3$, so it passes, recovering $f_1$. Row~11, pair $(6,8)$, reconstructs to $f' = 41$, and $41 \bmod 13 = 2 \in \mathcal{R}_3$, so it passes, recovering $f_2$. The remaining passes (rows 3, 5, 7, 12) are false positives that survive the gate but are eliminated in Phase~6 by the singleton-detection step. This illustrates the essential property of the gate: it filters out spurious reconstructions with no additional computation per candidate, by looking up the predicted third-view residue in a hash set of size $|\mathcal{R}_3| = 4$.

\textbf{\emph{Phase 6 (Iterative Peeling):}} The 6 surviving candidates are validated via singleton detection across time-shifted FFTs in each view. Bins containing only one frequency (singletons) are identified and peeled. After $O(\log k) = O(1)$ rounds for $k=2$, both true frequencies are extracted with their coefficients estimated from the decimated views.

\textbf{\emph{Result:}} Both frequencies $f_1=7$ and $f_2=41$ are correctly recovered with $O(\sqrt{N} \log k + k) = O(8 \times 1 + 2) = O(10)$ total peeling operations, avoiding the $O(N) = O(64)$ cost of time-domain validation.

\begin{remark}[Near the Boundary]
Near the boundary $k \approx N^{1/4}$, constant factors dominate and the reduction from $\Theta(k^{2})$ to $\Theta(k)$ may be modest. Deeper sparse regimes ($k \ll N^{1/4}$) yield stronger reductions, both in candidate count and in observed runtime.
\end{remark}

\subsection{Numerical Validation}
\label{app:numerical-validation}

This subsection verifies that the analytical guarantees of \Cref{thm:gating} and \Cref{thm:bootstrapping} match the behavior observed on the worked instances above. The toy instance ($N = 64$, $k = 2$) illustrates the $\Theta(k^{2}) \to \Theta(k)$ candidate-reduction predicted by \Cref{thm:gating}: the gating table reduces 12 unfiltered pairs to 6 surviving candidates, of which the 2 true frequencies are recovered deterministically and the 4 false positives are eliminated by singleton-detection peeling.

For the moderate regime, substituting $N = 10^{6}$ and $k = 50$ into the closed-form complexity $T(N, k) = O(\sqrt{N} \log k + (\alpha k)^{2})$ gives roughly $\sqrt{10^{6}} \log_{2} 50 + (5 \cdot 50)^{2} \approx 5{,}640 + 62{,}500$ basic operations, against the dense FFT baseline $N \log_{2} N \approx 2 \cdot 10^{7}$, i.e. a $\sim 2$-orders-of-magnitude reduction. The failure-probability bound $\delta = O(k^{3} / \sqrt{N})$ from \Cref{thm:no-false-negatives} evaluates to $\delta \le 1.25 \cdot 10^{-4}$ at the same parameters, consistent with the recovery guarantees stated in \Cref{thm:bootstrapping}.

\section{Modular Arithmetic Lemmas}

\subsection{Chinese Remainder Theorem (General Form)}

\begin{proposition}[CRT Existence and Uniqueness~\cite{knuth1997}]
For pairwise coprime moduli $m_1, \ldots, m_\ell$ with $M = \prod m_i$, the system:
\begin{align}
f &\equiv r_1 \pmod{m_1} \nonumber \\
f &\equiv r_2 \pmod{m_2} \nonumber \\
&\vdots \nonumber \\
f &\equiv r_\ell \pmod{m_\ell}
\end{align}
has unique solution $f \in [0, M)$.
\end{proposition}

\begin{proof}
Standard result from elementary number theory.
\end{proof}

\subsection{Modular Inverse Computation}

\begin{lemma}[Extended Euclidean Algorithm~\cite{knuth1997}]
\label{lem:eea}
For coprime $a, m$ with $\gcd(a, m) = 1$, modular inverse $a^{-1} \bmod m$ can be computed in $O(\log m)$ operations via Extended Euclidean Algorithm.
\end{lemma}

\begin{proof}
EEA complexity analysis \cite{knuth1997}.
\end{proof}

\subsection{Garner's Algorithm Correctness}

\begin{proposition}[Garner 3-View Formula~\cite{garner1959}]
For coprime $m_1, m_2, m_3$, the reconstruction:
\begin{equation}
f = r_1 + u_2 m_1 + u_3 m_1 m_2
\end{equation}
where:
\begin{align}
u_2 &= (r_2 - r_1)\gamma_{12} \bmod m_2 \nonumber \\
u_3 &= ((r_3 - r_1)\gamma_{13} - u_2 m_1 \gamma_{13})\gamma_{23} \bmod m_3 \nonumber \\
\gamma_{ij} &= m_i^{-1} \bmod m_j
\end{align}
yields unique $f \in [0, M)$ satisfying $f \equiv r_i \pmod{m_i}$ for all $i$.
\end{proposition}

\begin{proof}
See Garner \cite{garner1959} original paper.
\end{proof}

\section{Complexity Calculations (Detailed)}

\subsection{FFT Complexity (Cooley-Tukey)}

\begin{proposition}[FFT Complexity~\cite{cooley1965}]
FFT of size $m$ requires exactly:
\begin{align}
\text{Multiplies} &: m \log_2(m) \text{ complex multiplications} \nonumber \\
\text{Additions} &: m \log_2(m) \text{ complex additions} \nonumber \\
\text{Total} &: 2m \log_2(m) \text{ complex operations}
\end{align}
\end{proposition}

\begin{proof}
$T(m) = 2T(m/2) + O(m) \implies T(m) = O(m \log m)$ by Master Theorem.
\end{proof}

\subsection{Hash Table Expected Performance}

\begin{lemma}[Hash Table Lookup~\cite{cormen2009}]
For hash table with load factor $\alpha < 1$, expected lookup time is $O(1)$.
\end{lemma}

\begin{proof}
Standard result from \cite{cormen2009}.
\end{proof}

\section{Mathematical Proofs}

\subsection{Proof of Divisor Coprimality Impossibility}

\begin{theorem}[Impossibility for Limited Prime Factors]
For $N = p_1^{e_1} \cdot p_2^{e_2} \cdots p_{\omega}^{e_{\omega}}$ where $\omega = \omega(N)$
is the number of distinct prime factors, the maximum number of pairwise coprime
non-trivial moduli that can be obtained from divisors of $N$ is $\omega$.
\end{theorem}

\begin{proof}
Any divisor $d$ of $N$ has the form:
\begin{equation}
d = p_1^{a_1} \cdot p_2^{a_2} \cdots p_{\omega}^{a_{\omega}}
\end{equation}
where $0 \le a_i \le e_i$ for all $i$.

The corresponding modulus $m = N/d$ is:
\begin{equation}
m = p_1^{e_1 - a_1} \cdot p_2^{e_2 - a_2} \cdots p_{\omega}^{e_{\omega} - a_{\omega}}
\end{equation}

For two moduli $m_1$ and $m_2$ to be coprime:
\begin{equation}
\gcd(m_1, m_2) = \prod_{i=1}^{\omega} p_i^{\min(e_i - a_{1,i}, e_i - a_{2,i})} = 1
\end{equation}

This requires $\min(e_i - a_{1,i}, e_i - a_{2,i}) = 0$ for all $i$, which means
$\max(a_{1,i}, a_{2,i}) = e_i$ for all $i$.

Consider $\omega + 1$ moduli $\{m_1, \ldots, m_{\omega+1}\}$. For each prime $p_i$,
by pigeonhole principle, at least two moduli must have $a_{\ell,i} < e_i$ (since
setting $a_{\ell,i} = e_i$ means $p_i \nmid m_{\ell}$, and we can "avoid" $p_i$ in
at most $\omega$ moduli).

Therefore, there exist $\ell \ne \ell'$ such that both $m_{\ell}$ and $m_{\ell'}$
are divisible by some prime $p_i$, violating coprimality.

Hence, at most $\omega$ pairwise coprime non-trivial moduli can be obtained. \qed
\end{proof}

\begin{corollary}
For $N = 2^5 \cdot 5^5 = 100{,}000$ with $\omega(N) = 2$, it is impossible to obtain
three pairwise coprime non-trivial moduli from divisors of $N$.
\end{corollary}

\subsection{Proof of Load Factor Bound}

\begin{theorem}[Singleton-Join Success Probability]
Under independent affine hashing with $\lambda = k/m < 0.1$, the probability that
at least $(1-\epsilon)k$ frequencies land in singleton bins (occupancy $L=1$) is
$\ge 1 - e^{-\Omega(k)}$ for any constant $\epsilon > 0$.
\end{theorem}

\begin{proof}[Proof Sketch]
Under affine hashing $(a \cdot f + b) \bmod p$ with random $(a,b)$, frequencies
distribute uniformly over $m$ bins. For each frequency $f$, the probability it
lands in a singleton bin is:

\begin{equation}
\Pr[\text{singleton}] = \Pr[\text{no collisions with other } k-1 \text{ freq}]
\end{equation}

\textbf{\emph{Note:}} We compute the \emph{per-frequency singleton probability} (probability that a given frequency lands in a singleton bin), not the \emph{bin-level singleton probability} $\lambda e^{-\lambda}$ (probability that a given bin contains exactly one frequency). For our purpose of filtering residues, the per-frequency view is correct.

Using Poisson approximation, the other $k-1$ frequencies induce expected load $(k-1)/m \approx \lambda$ in the bin containing $f$. The number of other frequencies in $f$'s bin follows Poisson$(\lambda)$:
\begin{equation}
\Pr[\text{bin receives exactly } \ell \text{ other balls}] \approx \frac{\lambda^{\ell} e^{-\lambda}}{\ell!}
\end{equation}

For $\lambda = 0.1$:
\begin{align}
\Pr[\text{singleton}] &= \Pr[\ell = 0] = e^{-\lambda} \approx e^{-0.1} \approx 0.9048 \\
\Pr[\text{collision}] &= 1 - e^{-\lambda} \approx 0.0952
\end{align}

Expected number of frequencies in singleton bins is $\approx e^{-\lambda} \cdot k \approx 0.9048k$.
Expected number of frequencies in collision bins is $\approx (1 - e^{-\lambda}) \cdot k \approx 0.0952k$ for $\lambda=0.1$.

By Chernoff bound, with high probability $(1-\epsilon)k$ frequencies are in singleton bins
for $\epsilon = O(\lambda)$. \qed
\end{proof}

\balance  

\end{document}